\documentclass[12pt]{article}
\usepackage{latexsym}
\usepackage{amsmath}
\usepackage{amssymb}
\usepackage[english]{babel}
\usepackage{dsfont}
\usepackage{color}
\usepackage{mathrsfs}
\usepackage{graphicx}
\usepackage[utf8]{inputenc}
\usepackage{amsthm}
\usepackage[authoryear]{natbib}
\usepackage{dsfont}
\usepackage{color}
\usepackage{graphicx}
\usepackage{subcaption}
\usepackage{float}
\usepackage{varioref}
\usepackage{placeins}
\usepackage[table]{xcolor}
\usepackage{bbding}
\usepackage{changepage}
\newtheorem{theorem}{Theorem}[section]
\newtheorem{assumption}{Assumption}[section]

\newtheorem{definition}{Definition}[section]

\newtheorem{proposition}{Proposition}[section]
\newtheorem{property}{Property}[section]

\makeatother

\newcommand{\rv}{\textsc{\textbf{r}}}

\newcommand{\R}{\mathbb{R}}
\newcommand{\rr}{\boldsymbol{R}}
\newcommand{\N}{\mathbb{N}}

\newcommand{\E}{\mathbb{E}}
\newcommand{\X}{\boldsymbol{X}}
\newcommand{\M}{\mathcal{M}}
\newcommand{\e}{\mathbf{e}}

\newcommand{\F}{\boldsymbol{F}}

\newcommand{\x}{\mathbf{x}}

\newcommand{\cov}{\text{Cov}}

\newcommand{\z}{\mathbf{z}}
\newcommand{\f}{\mathbf{f}}

\newcommand{\s}{\mathbf{s}}
\newcommand{\bb}{\mathbf{b}}
\newcommand{\B}{\boldsymbol{B}}
\newcommand{\Z}{\boldsymbol{Z}}
\newcommand{\V}{\boldsymbol{V}}
\def\independenT#1#2{\mathrel{\rlap{$#1#2$}\mkern2mu{#1#2}}}
\newcommand\independent{\protect\mathpalette{\protect\independenT}{\perp}}

\usepackage[usenames,dvipsnames]{pstricks}
\usepackage{epsfig}
\usepackage{pst-grad} % For gradients
\usepackage{pst-plot} % For axes

\usepackage[titlenotnumbered,ruled,vlined,noend,linesnumbered]{algorithm2e}
\usepackage{array,multirow}

\begin{document}

\title{Principal Component Analysis:\\ A Generalized Gini Approach}

\author{Arthur Charpentier \\
UQAM \\  \and St\'{e}phane Mussard\thanks{$^{{}}$Universit\'{e}  de N\^{i}mes -- e-mail: stephane.mussard@unimes.fr,
Research fellow \textsc{Gr\'{e}di} University of Sherbrooke, and \textsc{Liser} Luxembourg.}\\
\textsc{Chrome}\\
Universit\'{e} de N\^{i}mes \and 
T\'{e}a Ouraga \\ \textsc{Chrome} \\ Universit\'{e} de N\^{i}mes} 

\date{}
\maketitle

\begin{abstract}
A principal component analysis based on the generalized Gini correlation index is proposed (Gini PCA). The Gini PCA generalizes the standard PCA based on the variance. It is shown, in the Gaussian case, that the standard PCA is equivalent to the Gini PCA. It is also proven that the dimensionality reduction based on the generalized Gini correlation matrix, that relies on city-block distances, is robust to outliers. Monte Carlo simulations and an application on cars data (with outliers) show the robustness of the Gini PCA and provide different interpretations of the results compared with the variance PCA.  
\end{abstract}

\textbf{keywords} : Generalized Gini ; PCA ; Robustness.

%\textbf{JEL Codes} : C38.

\break

\section{Introduction} 

This late decade, a line of research has been developed and focused on the Gini methodology, see \citet{Yitzhaki13} for a general review of different Gini approaches applied in Statistics and in Econometrics.\footnote{See \citet{giorgi} for an overview of the "Gini methodology".} Among the Gini tools, the Gini regression has received a large audience since the Gini regression initiated by \citet{ref9}. Gini regressions have been generalized by \citet{Yitzhaki13} in different areas and particularly in time series analysis. \citet{ref13} and \citet{ref2} investigated ARMA processes with an identification and an estimation procedure based on Gini autocovariance functions. This robust Gini approach has been shown to be relevant to heavy tailed distributions such as Pareto processes. Also, \citet{shelef16} proposed a unit root test based on Gini regressions to deal with outlying observations in the data. \medskip

In parallel to the above literature, a second line of research on multidimensional Gini indices arose. This literature paved the way on the valuation of inequality about multiple commodities or dimensions such as education, health, income, etc., that is, to find a real-valued function that quantifies the inequality between the households of a population over each dimension, see among others, \cite{List}, \cite{Gajdos}, \cite{Decancq}. More recently, \cite{Banerjee} shows that it is possible to construct multidimensional Gini indices by exploring the projection of the data in reduced subspaces based on the Euclidean norm. Accordingly, some notions of linear algebra have been increasingly included in the axiomatization of multidimensional Gini indices.\medskip

In this paper, in the same vein as in the second line of research mentioned above, we start from the recognition that linear algebra may be closely related to the maximum level of inequality that arises in a given dimension. In data analysis, the variance maximization is mainly used to further analyze projected data in reduced subspaces. The variance criterion implies many problems since it captures a very precise notion of dispersion, which does not always match some basic properties satisfied by variability measures such as the Gini index. Such a property may be, for example, an invariance condition postulating that a dispersion measure remains constant when the data are transformed by monotonic maps.\footnote{See \citet{furman} for the link between variability (risk) measures and the Gini correlation index.} Another property typically related to the Gini index is its robustness to outlying observations, see e.g. \cite{Olkin} %and \cite{Yitzhaki1992}
in the case of linear regressions. Accordingly, it seems natural to analyze multidimensional dispersion with the Gini index, instead of the variance, in order to provide a Principal Components Analysis (PCA) in a Gini sense (Gini PCA). \medskip

In the field of PCA, \cite{Baccini} and \cite{Korhonen} are among the first authors dealing with a $\ell_1$-norm PCA framework. Their idea was to robustify the standard PCA by means of the Gini Mean Difference metric introduced by \cite{Gini12}, which is a city-block distance measure of variability. The authors employ the Gini Mean Difference as an estimator of the standard deviation of each variable before running the singular value decomposition leading to a robust PCA. In the same vein, \cite{Ding2006} make use of a rotational $\ell_1$ norm PCA to robustify the variance-covariance matrix in such a way that the PCA is rotational invariant. Recent PCAs derive latent variables thanks to regressions based on \emph{elastic net} (a $\ell_1$ regularization) that improves the quality of the regression curve estimation, see \cite{Zou}. %Partial Least Squares (PLS) regressions on $\ell_1$ metric spaces have also been employed in order to project the data with a better adjustment in the presence of extreme values, see the Gini-PLS regression of \cite{Mussard}. \medskip

\medskip
 
In this paper, it is shown that the variance may be seen as an inappropriate criterion for dimensionality reduction in the case of data contamination or outlying observations. A generalized Gini PCA is investigated by means of Gini correlations matrices. These matrices contain generalized Gini correlation coefficients (see \cite{Yitzhaki03}) based on the Gini covariance operator introduced by \cite{Schechtman87} and \cite{Schechtman03}. The generalized Gini correlation coefficients are: (i) bounded, (ii) invariant to monotonic transformations, (iii) and symmetric whenever the variables are exchangeable. It is shown that the standard PCA is equivalent to the Gini PCA when the variables are Gaussians. Also, it is shown that the generalized Gini PCA may be realized either in the space of the variables or in the space of the observations. In each case, some statistics are proposed to perform some interpretations of the variables and of the observations (absolute and relative contributions). To be precise, an $U$-statistics test is introduced to test for the significance of the correlations between the axes of the new subspace and the variables in order to assess their significance. Monte Carlo simulations are performed in order to show the superiority of the Gini PCA compared with the usual PCA when outlying observations contaminate the data. Finally, with the aid of the well-known cars data, which contain outliers, it is shown that the generalized Gini PCA leads to different results compared with the usual PCA. 

\medskip

The outline of the paper is as follows. Section \ref{s2} sets the notations and presents some $\ell_2$ norm approaches of PCA. Section \ref{s3} reviews the Gini-covariance operator.  Section \ref{s4} is devoted to the generalized Gini PCA.  Section \ref{s5} focuses on the interpretation of the Gini PCA.  Sections \ref{s6} and \ref{s7} present some Monte Carlo simulations and applications, respectively.

\section{Motivations for the use of Gini PCA}\label{s2}

In this Section, the notations are set. Then, some assumptions are imposed and some $\ell_2$-norm PCA techniques are reviewed in order to motivate the employ of the Gini PCA.  

\subsection{Notations and definitions}

Let $\N^{*}$ be the set of integers and $\R$ $[\R_{++}]$ the set of [positive] real numbers. Let $\M$ be the set of all $N\times K$ matrix $\X=[x_{ik}]$ that describes $N$ observations on $K$ dimensions such that $N \gg K$, with elements $x_{ik}\in \R$, and $\mathbb{I}_n$ the $n\times n$ identity matrix. The $N\times 1$ vectors representing each variable are expressed as $\x_{\cdot k}$, for all $k \in \{1,\ldots,K\}$ and we assume that $\x_{\cdot k} \neq c\mathbf{1}_{N}$, with $c$ a real constant and $\mathbf{1}_{N}$ a $N$-dimensional column vector of ones. The $K\times 1$ vectors representing each observation $i$ (the transposed $i$th line of $\X$) are expressed as $\x_{i\cdot}$, for all $i \in \{1,\ldots,N\}$. It is assumed that $\x_{\cdot k}$ is the realization of the random variable $X_k$, with cumulative distribution function $F_k$. The arithmetic mean of each column (line) of the matrix $\X$ is given by $\bar{\x}_{\cdot k}$ ($\bar{\x}_{i\cdot}$). The cardinal of set $A$ is denoted $\# \{A\}$. The $\ell_1$ norm, for any given real vector $\x$, is $\| \mathbf x \|_1 = \sum_{k=1}^K|x_{\cdot k}|$, whereas the $\ell_2$ norm is $\| \x \|_2 =(\sum_{k=1}^Kx_{\cdot k}^2)^{1/2}$.

\begin{assumption}\
The random variables $X_k$ are such that $\mathbb E[|X_k|] < \infty$ for all $k \in \{1,\ldots,K\}$, but no assumption is made on the second moments (that may not exist).
\end{assumption}

This assumption imposes less structure compared with the classical PCA in which the existence of the second moments are necessary, as can be seen in the next subsection.

\subsection{Variants of PCA based on the $\ell_2$ norm}

The classical formulation of the PCA, to obtain the first component, can be obtained by solving
\begin{equation}\label{eq:PCA0}
\omega_1^*\in\text{argmax}\left\lbrace\text{Var}[\X\omega]\right\rbrace
\text{ subject to }\|\omega\|_2^2=\omega^{\top}\omega=1,\end{equation}
or equivalently
\begin{equation}\label{eq:PCA}
\omega_1^*\in\text{argmax}\left\lbrace\omega^{\top}\boldsymbol{\Sigma}\omega\right\rbrace
\text{ subject to }\|\omega\|_2^2=\omega^{\top}\omega=1,
\end{equation}
where $\omega\in\mathbb{R}^K$, and $\boldsymbol{\Sigma}$ is the (symmetric positive semi-definite) $K\times K$ sample covariance matrix.
\cite{Mardia} suggest to write$$
\omega_1^*\in\text{argmax}\left\lbrace\sum_{j=1}^K\text{Var}[\x_{\cdot, j}]\cdot\text{Cor}[\x_{\cdot,j},\X\omega]\right\rbrace
\text{ subject to }\|\omega\|_2^2=\omega^{\top}\omega=1.$$
With scaled variables\footnote{In most cases, PCA is performed on scaled (and centered) variables, otherwise variables with large scales might alter interpretations. Thus, it will make sense, later on, to assume that components of $\X$ have identical distributions. At least the first two moments will be equal.\label{foot3}} (i.e. $\text{Var}[\x_{\cdot, j}]=1$, $\forall j$)
\begin{equation}\label{eq:PCA:norm}
\omega_1^*\in\text{argmax}\left\lbrace\sum_{j=1}^K \text{Cor}[\x_{\cdot,j},\X\omega]\right\rbrace
\text{ subject to }\|\omega\|_2^2=\omega^{\top}\omega=1.
\end{equation}
Then a Principal Component Pursuit can start: we consider the `residuals', $\X_{(1)}=\X-\X\omega_1^*\omega_1^{*\top}$, its covariance matrix $\boldsymbol{\Sigma}_{(1)}$, and we solve
$$
\omega_2^*\in\text{argmax}\left\lbrace\omega^{\top}\boldsymbol{\Sigma}_{(1)}\omega\right\rbrace
\text{ subject to }\|\omega\|_2^2=\omega^{\top}\omega=1.$$
The part $\X\omega_1^*\omega_1^{*\top}$ is actually a constraint that we add to ensure the orthogonality of the two first components. This problem is equivalent to finding the maxima of $\text{Var}[\X\omega]$ subject to $\|\omega\|_2^2=1$ and $\omega\perp\omega_1^*$.
This idea is also called Hotelling (or Wielandt) deflation technique. On the $k$-th iteration, we extract the leading eigenvector
$$
\omega_k^*\in\text{argmax}\left\lbrace\omega^{\top}\boldsymbol{\Sigma}_{(k-1)}\omega\right\rbrace
\text{ subject to }\|\omega\|_2^2=\omega^{\top}\omega=1,$$
where $\boldsymbol{\Sigma}_{(k-1)}=\boldsymbol{\Sigma}_{(k-2)}-\omega_{k-1}^{*}\omega_{k-1}^{*\top}\boldsymbol{\Sigma}_{(k-1)}\omega_{k-1}^{*}\omega_{k-1}^{*\top}$ (see e.g. \cite{Saad}). Note that, following \cite{Hotelling} and \cite{Eckart}, that it is also possible to write this problem as
$$
\min\left\lbrace\|\X-\tilde{\X}\|_\star\right\rbrace
\text{ subject to }\text{rank}[\tilde{\X}]\leq k
$$
where $\|\cdot\|_\star$ denotes the nuclear norm of a matrix (\emph{i.e.} the sum of its singular values)\footnote{but other norms have also been considered in statistical literature, such as the Froebenius norm in the Eckart-Young theorem, or the maximum of singular values -- also called $2$-(induced)-norm.}. 

One extension, introduced in \cite{dAspremont}, was to add a constraint based on the cardinality of $\omega$ (also called $\ell_0$ norm) corresponding to the number of non-zero coefficients of $\omega$. The penalized objective function is then
$$
\max\left\lbrace\omega^{\top}\boldsymbol{\Sigma}\omega-\lambda\text{card}[\omega]\right\rbrace
\text{ subject to }\|\omega\|_2^2=\omega^{\top}\omega=1,
$$
for some $\lambda>0$. This is called {\em sparse PCA}, and can be related to sparse regression, introduced in \cite{Tibshirani}. But as pointed out in \cite{Mackey}, interpretation is not easy and the components obtained are not orthogonal. \cite{Gorban} considered an extension to nonlinear Principal Manifolds to take into account nonlinearities.

Another direction for extensions was to consider Robust Principal Component Analysis. \cite{Candes} suggested an approach based on the fact that principal component pursuit can be obtained by solving 
$$
\min\left\lbrace\|\X-\tilde{\X}\|_\star+\lambda\|\tilde{\X}\|_1\right\rbrace.$$
But other methods were also considered to obtain Robust PCA. A natural `scale-free' version is obtained by considering a rank matrix instead of $\X$. This is also called `ordinal' PCA in the literature, see \cite{Korhonen}. The first `ordinal' component is
\begin{equation}
\omega_1^*\in\text{argmax}\left\lbrace\sum_{j=1}^K\mathcal{R}[\x_{\cdot,j},\X\omega]\right\rbrace
\text{ subject to }\|\omega\|_2^2=\omega^{\top}\omega=1
\end{equation}
where $\mathcal{R}$ denotes some rank based correlation, \emph{e.g.} Spearman's rank correlation, as an extention of Equation (\ref{eq:PCA:norm}). So, quite naturally, one possible extension of Equation (\ref{eq:PCA}) would be
$$
\omega_1^*\in\text{argmax}\left\lbrace\omega^{\top}\mathcal{R}[\X]\omega\right\rbrace
\text{ subject to }\|\omega\|_2^2=\omega^{\top}\omega=1$$
where $\mathcal{R}[\X]$ denotes Spearman's rank correlation. In this section, instead of using Pearson's correlation (as in Equation (\ref{eq:PCA}) when the variables are scaled) or Spearman's (as in this ordinal PCA), we will consider the multidimensional Gini correlation based on the $h$-covariance operator.

\section{Geometry of Gini PCA: Gini-Covariance Operators}\label{s3}

The first PCA was introduced by \cite{Pearson}, projecting $\boldsymbol{X}$ onto the eigenvectors of its covariance matrix, and observing that the variances of those projections are the corresponding eigenvalues. One of the key property is that $\X^\top\X$ is a positive matrix. Most statistical properties of PCAs (see \cite{FluryRiedwyl} or \cite{Anderson}) are obtained under Gaussian assumptions. Furthermore, geometric properties can be obtained using the fact that the covariance defines an inner product on the subspace of random variables with finite second moment (up to a translation, \emph{i.e.} we identify any two that differ by a constant). 

We will discuss in this section the properties of the Gini Covariance operator with the special case of Gaussian random variables, and the property of the Gini correlation matrix that will be used in the next Section for the Gini PCA. 

\subsection{The Gini-covariance operator}\label{section}

In this section, $\boldsymbol{X}=(X_1,\cdots,X_K)$ denotes a random vector. 
The covariance matrix between $\boldsymbol{X}$ and $\boldsymbol{Y}$, two random vectors, is defined as the inner product between centered versions of the vectors, 
\begin{equation}\label{eq:innerprod}
\langle \boldsymbol{X},\boldsymbol{Y}\rangle=\text{Cov}(\boldsymbol{X},\boldsymbol{Y})=\mathbb{E}[(\boldsymbol{X}-\mathbb{E}[\boldsymbol{X}])(\boldsymbol{Y}-\mathbb{E}[\boldsymbol{Y}])^{\text{\sffamily T}}].
\end{equation}
Hence, it is the matrix where elements are regular covariances between components of the vectors, $\text{Cov}(\boldsymbol{X},\boldsymbol{Y})=[\text{Cov}(X_i,Y_j)]$. It is the upper-right block of the covariance matrix of $(\boldsymbol{X},\boldsymbol{Y})$. Note that $\text{Cov}(\boldsymbol{X},\boldsymbol{X})$ is the standard variance-covariance matrix of vector $\boldsymbol{X}$.

\begin{definition}
Let $\boldsymbol{X}=(X_1,\cdots,X_{K})$ be collections of $K$ identically distributed random variables. Let $h:\mathbb{R}\rightarrow\mathbb{R}$ denote a non-decreasing function. Let $h(\boldsymbol{X})$ denote the random vector $(h(X_1),\cdots,h(X_K))$, and assume that each component has a finite variance. Then, operator $\Gamma C_{h}(\boldsymbol{X})=\emph{Cov}(\boldsymbol{X},h(\boldsymbol{X}))$ is called $h$-Gini covariance matrix. 
\end{definition}

Since $h$ is a non-decreasing mapping, then $\boldsymbol{X}$ and $h(\boldsymbol{X})$ are componentwise comonotonic random vectors. Assuming that components of $\boldsymbol{X}$ are identically distributed is a reasonable assumption in the context of scaled (and centered) PCA, as discussed in footnote \ref{foot3}. Nevertheless, a stronger technical assumption will be necessary: pairwise-exchangeability.

\begin{definition}
$\boldsymbol{X}$ is said to be pairwise-exchangeable if for all pair $(i,j)\in\{1,\cdots,K\}^2$, $(X_i,X_j)$ is exchangeable, in the sense that $(X_i,X_j)\overset{\mathcal{L}}{=}(X_j,X_i)$.
\end{definition}

Pairwise-exchangeability is a stronger concept than having only one vector with identically distributed components, and a weaker concept than (full) exchangeability. In the Gaussian case where $h(X_k) = \Phi(X_k)$ with $\Phi(X_k)$ being the normal cdf of $X_k$ for all $k=1,\ldots,K$, pairwise-exchangeability is equivalent to components identically distributed.

\begin{proposition}
If $\boldsymbol{X}$ is a Gaussian vector with identically distributed components, then $\boldsymbol{X}$ is pairwise-exchangeable.
\end{proposition}

\begin{proof}
For simplicity, assume that components of $\boldsymbol{X}$ are $\mathcal{N}(0,1)$ random variables, then $\boldsymbol{X}\sim\mathcal{N}(\boldsymbol{0},\boldsymbol{\rho})$ where $\boldsymbol{\rho}$ is a correlation matrix. In that case
$$
\left[
\begin{matrix}
X_i \\
X_j
\end{matrix}
\right]\sim\mathcal{N}\left(
\left[
\begin{matrix}
0 \\
0
\end{matrix}
\right],
\left[
\begin{matrix}
1 & \rho_{ij} \\
\rho_{ji} & 1 \\
\end{matrix}
\right]
\right),$$
with Pearson correlation $\rho_{ij}=\rho_{ji}$, thus $(X_i,X_j)$ is exchangeable.
\end{proof}

Let us now introduce the Gini-covariance. \cite{Gini12} introduced the Gini mean difference operator $\Delta$, defined as: 
\begin{equation}\label{eq:gmd}
\Delta(X)=\mathbb{E}(|X_1-X_2|)\text{ where }X_1,X_2\sim X,\text{ and }X_1\independent X_2, 
\end{equation}
for some random variable $X$ (or more specifically for some distribution $F$ with $X\sim F$, because this operator is law invariant). One can rewrite:
$$
\Delta(X)=4\text{Cov}[X,F(X)]=\frac{1}{3}\frac{\text{Cov}[X,F(X)]}{\text{Cov}[F(X),F(X)]}
$$
where the term on the right is interpreted as the slope of the regression curve of the observed variable $X$ and its `ranks' (up to a scaling coefficient). Thus, the Gini-covariance is obtained when the function $h$ is equal to the cumulative distribution function of the second term, see \cite{Schechtman87}.  

\begin{definition}
Let $\boldsymbol{X}=(X_1,\cdots,X_{K})$ be a collection of $K$ identically distributed random variables, with cumulative distribution function $F$. Then, the Gini covariance is $\Gamma C_F(\boldsymbol{X})=\emph{\text{Cov}}(\boldsymbol{X},F(\boldsymbol{X}))$.
\end{definition}

On this basis, it is possible to show that the Gini covariance matrix is a positive semi-definite matrix. 

\begin{theorem}
Let $Z \sim \mathcal{N}(0,1)$. If $\boldsymbol{X}$ represents identically distributed Gaussian random variables, with distribution $\mathcal{N}(\mu,\sigma^2)$, then the two following assertions hold:

\emph{(i)} $\Gamma C_F(\boldsymbol{X})=\sigma^{-1} \emph{Cov}(Z,\Phi(Z)) \emph{Var}(\boldsymbol{X}).$

\emph{(ii)} $\Gamma C_F(\boldsymbol{X})$ is a positive-semi definite matrix. 
\end{theorem}
\begin{proof}
(i) In the Gaussian case, if $h$ is the cumulative distribution function of the $X_k$'s, then $\text{Cov}(X_k,h(X_\ell))=r\sigma\cdot\text{Cov}(Z,\Phi(Z))$, where $\Phi$ is the normal cdf, see \cite{Yitzhaki13}, Chapter 3.  %Numerically, one can prove that
%$$
%\lambda:=\text{cov}(Z,\Phi(Z))\approx 0.2821.
%$$
Observe that $\text{Cov}(X_k,h(X_k))=\sigma\cdot \text{Cov}(Z,\Phi(Z))$, if $h$ is the cdf of $X_k$. Thus, $\lambda := \cov(Z,\Phi(Z))$ yields: 
$$\Gamma C_F((X_k,X_\ell))=\lambda \left[
\begin{matrix}
\sigma & \rho \sigma \\
\rho\sigma  & \sigma \\
\end{matrix}
\right]=
\lambda \sigma \left[
\begin{matrix}
1 & \rho \\
\rho & 1\\
\end{matrix}
\right]=\frac{\lambda}{\sigma}
\text{Var}((X_k,X_\ell)).$$

\noindent (ii) We have $\text{Cov}(Z,\Phi(Z)) \geq 0$, then it follows that $C_F((X_k,X_\ell)) \geq 0$: 
$$
\x^{\top}C_F(\boldsymbol{X}) \x = \x^{\top}\frac{\cov(Z,\Phi(Z))}{\sigma}\text{Var}(\boldsymbol{X}) \x \geq 0,
$$
which ends the proof.
\end{proof}

Note that $\Gamma C_F(\boldsymbol{X})=\text{Cov}(\boldsymbol{X},-\overline{F}(\boldsymbol{X}))=\Gamma C_{-\overline{F}}(\boldsymbol{X})$, where $\overline{F}$ denotes the survival distribution function.

\begin{definition}
Let $\boldsymbol{X}=(X_1,\cdots,X_{K})$ be a collection of $K$ identically distributed random variables, with survival distribution function $\overline{F}$. Then, the generalized Gini covariance is $G\Gamma C_{\nu}(\boldsymbol{X})=\Gamma C_{-\overline{F}^{\nu-1}}(\boldsymbol{X})=\text{\em Cov}(\boldsymbol{X},-\overline{F}^{\nu-1}(\boldsymbol{X}))$, for $\nu>1$.
\end{definition}

This operator is related to the one introduced in \cite{Schechtman03}, called generalized Gini mean difference $GMD_\nu$ operator. More precisely, an estimator of the generalized Gini mean difference is given by: 
\[
GMD_\nu(\x_{\cdot\ell},\x_{\cdot k}) := -\frac{2}{N-1}\nu \cov(\x_{\cdot\ell},\rv_{\x_{\cdot k}}^{\nu-1}) ,  \ \nu > 1,
\]
where $\rv_{\x_{\cdot k}} = (R(x_{1k}),\ldots,R(x_{nk}))$ is the decumulative rank vector of $\x_{\cdot k}$, that is, the vector that assigns the smallest value (1) to the greatest observation $x_{ik}$, and so on. The rank of observation $i$ with respect to variable $k$ is: 
\[
R(x_{ik}) :=
\left\{ \begin{array}{ll}
N+1- \#\{x \leq x_{ik} \} & \text{if no ties} \\
N+1-\frac{1}{p}\sum_{i=1}^p \#\{ x \leq x_{ik}  \} & \text{if $p$ ties $x_{ik}$.}
\end{array}
\right.
\]
Hence $GMD_\nu(\x_{\cdot\ell},\x_{\cdot k})$ is the empirical version of
$$
2\nu \Gamma C_\nu(X_{\ell},X_{ k}) := -2\nu\cov\big(X_\ell,\overline{F}_k(X_k)^{\nu-1}\big).
$$
The index $GMD_\nu$ is a generalized version of the $GMD_2$ proposed earlier by \cite{Schechtman87}, and can also be written as:
$$
GMD_2(X_{k},X_{k}) = 4\cov\big(X_k,F_k(X_k)\big) = \Delta(X_k).
$$
When $k=\ell$, $GMD_\nu$ represents the variability of the variable $\x_{\cdot k}$ itself. Focus is put on the lower tail of the distribution $\x_{\cdot k}$ whenever $\nu \to \infty$, the approach is said to be max-min in the sense that $GMD_\nu$ inflates the minimum value of the distribution. On the contrary, whenever $\nu \to 0$, the approach is said to be max-max, in this case focus is put on the upper tail of the distribution $\x_{\cdot k}$. As mentioned in \cite{Yitzhaki13}, the case $\nu < 1$ does not entail simple interpretations, thereby the parameter $\nu$ is used to be set as $\nu > 1$ in empirical applications.\footnote{In risk analysis $\nu \in (0,1)$ denotes risk lover decision makers (max-max approach), whereas $\nu > 1$ stands for risk averse decision makers, and $\nu \to \infty$ extreme risk aversion (max-min approach).} 

Note that even if $X_k$ and $X_\ell$ have the same distribution, we might have $GMD_\nu(X_k,X_\ell)\neq GMD_\nu(X_\ell,X_k)$, as shown on the example of Figure \ref{Fig:dessin-echangeable}. In that case $\mathbb{E}[X_kh(X_\ell)]\neq \mathbb{E}[X_\ell h(X_k)]$ if $h(2)\neq 2h(1)$ (this property is nevertheless valid if $h$ is linear). We would have $GMD_\nu(X_k,X_\ell)= GMD_\nu(X_\ell,X_k)$ when $X_k$ and $X_\ell$ are exchangeable. But since generally $GMD_\nu$ is not symmetric, we have for $\x_{\cdot k}$ being not a monotonic transformation of $\x_{\cdot\ell}$ and $\nu > 1$, $GMD_\nu(\x_{\cdot k},\x_{\cdot\ell})\neq GMD_\nu(\x_{\cdot\ell},\x_{\cdot k})$.

\begin{figure}
\centering\includegraphics[width=.4\textwidth]{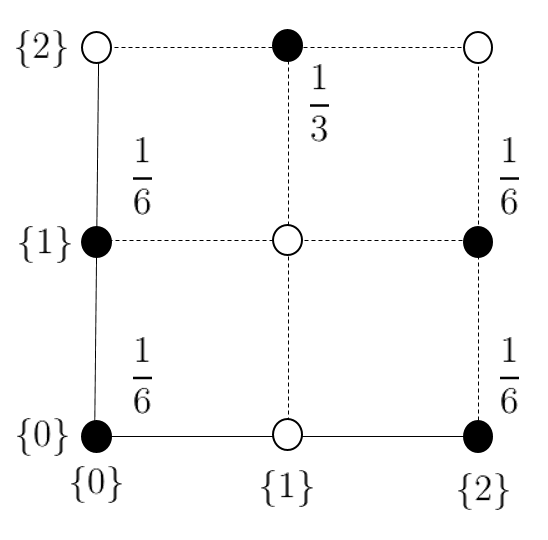}
\caption{Joint distribution of a random pair \ensuremath{(X_k,X_{\ell})} such that \ensuremath{\mathbb{E}[X_k h(X_{\ell})]\neq \mathbb{E}[X_{\ell} h(X_k)]}, with non-exchangeable components \ensuremath{X_k\overset{\mathcal{L}}{=}X_{\ell}}.}\label{Fig:dessin-echangeable}
\end{figure}

\subsection{Generalized Gini correlation}

In this section, $\X$ is a matrix in $\M$. The Gini correlation coefficient ($G$-correlation frown now on), is a normalized $GMD_\nu$ index such that for all $\nu > 1$, see \cite{Schechtman03},
\[
GC_\nu(\x_{\cdot\ell},\x_{\cdot k}) := \frac{GMD_\nu(\x_{\cdot\ell},\x_{\cdot k})}{GMD_\nu(\x_{\cdot\ell},\x_{\cdot\ell})} \ ; \ GC_\nu(\x_{\cdot k},\x_{\cdot\ell}) := \frac{GMD_\nu(\x_{\cdot k},\x_{\cdot\ell})}{GMD_\nu(\x_{\cdot k},\x_{\cdot k})},
\]
with $GC_\nu(\x_{\cdot k},\x_{\cdot k}) = 1$ and $GMD_\nu(\x_{\cdot k},\x_{\cdot k}) \neq 0$, for all $k,\ell=1,\ldots,K$. Following \cite{Schechtman03}, the $G$-correlation is well-suited for the measurement of correlations between non-normal distributions or in the presence of outlying observations in the sample. 

\begin{property}\label{prop1b} -- \textbf{\emph{Schechtman and Yitzhaki (2013):}}

\noindent \emph{(i)} $GC_\nu(\x_{\cdot\ell},\x_{\cdot k}) \leq 1$.

\noindent\emph{(ii)} If the variables $\x_{\cdot\ell}$ and $\x_{\cdot k}$ are independent, for all $k\neq \ell$, then $GC_\nu(\x_{\cdot\ell},\x_{\cdot k}) = GC_\nu(\x_{\cdot k},\x_{\cdot\ell}) =0$.

\noindent\emph{(iii)} For any given monotonic increasing transformation $\varphi$, $GC_\nu(\x_{\cdot\ell},\varphi(\x_{\cdot k})) = GC_\nu(\x_{\cdot\ell},\x_{\cdot k})$.

\noindent\emph{(iv)} If $(\x_{\cdot\ell},\x_{\cdot k})$ have a bivariate normal distribution with Pearson correlation $\rho$, then $GC_\nu(\x_{\cdot\ell},\x_{\cdot k}) = GC_\nu(\x_{\cdot k}, \x_{\cdot\ell}) = \rho$.

\noindent\emph{(v)} If $\x_{\cdot k}$ and $\x_{\cdot\ell}$ are exchangeable up to a linear transformation, then $GC_\nu(\x_{\cdot\ell},\x_{\cdot k}) = GC_\nu(\x_{\cdot k},\x_{\cdot\ell})$.
\end{property}

Whenever $\nu \to 1$, the variability of the variables is attenuated so that $GMD_\nu$ tends to zero (even if the variables exhibit a strong variance). The choice of $\nu$ is interesting to perform generalized Gini PCA with various values of $\nu$ in order to robustify the results of the PCA, since the standard PCA (based on the variance) is potentially of bad quality if outlying observations drastically affect the sample. 

A $G$-correlation matrix is proposed to analyze the data into a new vector space. Following Property \ref{prop1b} (iv), it is possible to rescale the variables $\x_{\cdot \ell}$ thanks to a linear transformation, then the matrix of standardized observation is,
\begin{equation}\label{standard}
\Z\equiv [z_{i\ell}] := \left[ \frac{x_{i\ell} - \bar{\x}_{\cdot \ell}}{GMD_\nu(\x_{\cdot\ell},\x_{\cdot\ell})}\right].
\end{equation}
The variable $z_{i\ell}$ is a real number without dimension. The variables $\x_{\cdot k}$ are rescaled such that their Gini variability is equal to unity. Now, we define the $N\times K$ matrix of decumulative centered rank vectors of $\Z$, which are the same compared with those of $\X$: 
\[
\rr_\z^c \equiv [R^c(z_{i\ell})] := [R(z_{i\ell})^{\nu-1} - \bar{\rv}^{\nu-1}_{\z_{\cdot\ell}}] = [R(x_{i\ell})^{\nu-1} - \bar{\rv}_{\x_{\cdot\ell}}^{\nu-1}].  
\]
Note that the last equality holds since the standardization (\ref{standard}) is a strictly increasing affine transformation.\footnote{By definition $GMD_\nu(\x_{\cdot \ell},\x_{\cdot \ell}) \geq 0$ for all $\ell=1,\ldots,K$. As we impose that $\x_{\cdot \ell} \neq c\mathbf{1}_N$, the condition becomes $GMD_\nu(\x_{\cdot \ell},\x_{\cdot \ell}) >0$.} The $K\times K$ matrix containing all $G$-correlation indices between all couples of variables $\z_{\cdot k}$ and $\z_{\cdot \ell}$, for all $k,\ell=1,\ldots,K$ is expressed as:
\begin{equation}\notag
GC_\nu(\Z) := -\frac{2\nu}{N(N-1)}\Z^{\top}\rr_\z^c.
\end{equation}
Indeed, if $GMD_\nu(\Z) \equiv [GMD_\nu(\z_{\cdot k},\z_{\cdot \ell})]$, then we get the following.

\begin{proposition}\label{prop1} For each standardized matrix $\Z$ defined in \emph{(\ref{standard})}, the following relations hold:
\begin{align}\label{GMD=GC}
& GMD_\nu(\Z) = GC_\nu(\X) = GC_\nu(\Z).\\
& GMD_\nu(\z_{\cdot k},\z_{\cdot k}) = 1, \ \forall k=1,\ldots,K.\label{GMD=1}
\end{align}
\end{proposition}

\begin{proof}
We have $GMD_\nu(\Z) \equiv [GMD_\nu(\z_{\cdot k},\z_{\cdot \ell})]$ being a $K\times K$ matrix. The extra diagonal terms may be rewritten as,
\begin{align}
GMD_\nu(\z_{\cdot k},\z_{\cdot \ell}) &=  -\frac{2}{N-1}\nu \cov(\z_{\cdot k},\rv_{\z_{\cdot \ell}}^{\nu-1}) \notag \\
&= -\frac{2}{N-1}\nu \cov \left( \frac{\x_{\cdot k} - \bar{\x}_{\cdot k}\mathbf{1}_N }{GMD_\nu(\x_{\cdot k},\x_{\cdot k})},\rv_{\z_{\cdot \ell}}^{\nu-1} \right) \notag \\
& = -\frac{2}{GMD_\nu(\x_{\cdot k},\x_{\cdot k})}\left[ \frac{\nu \cov( \x_{\cdot k}, \rv_{\z_{\cdot \ell}}^{\nu-1})}{N-1} - \frac{\nu \cov( \bar{\x}_{\cdot k}\mathbf{1}_N,\rv_{\z_{\cdot \ell}}^{\nu-1})}{N-1}\right] \notag \\
&= \frac{GMD_\nu(\x_{\cdot k},\x_{\cdot \ell})}{GMD_\nu(\x_{\cdot k},\x_{\cdot k})} = GC_\nu(\x_{\cdot k},\x_{\cdot \ell}).\notag
\end{align}
Finally, using the same approach as before, we get:
\begin{align}
GMD_\nu(\z_{\cdot k},\z_{\cdot k}) &=  -\frac{2}{N-1} \nu\cov(\z_{\cdot k},\rv_{\z_{\cdot k}}^{\nu-1}) \notag \\
&= \frac{GMD_\nu(\x_{\cdot k},\x_{\cdot k})}{GMD_\nu(\x_{\cdot k},\x_{\cdot k})} \notag \\ 
&= GC_\nu(\x_{\cdot k},\x_{\cdot k})=1.\notag
\end{align}
By Property \ref{prop1} (iv), since $\rv_{\x_{\cdot k}} = \rv_{\z_{\cdot k}}$, then $GC_\nu(\x_{\cdot k},\x_{\cdot \ell}) = GC_\nu(\z_{\cdot k},\z_{\cdot \ell})$. Thus,
\[
GMD_\nu(\Z) = -\frac{2\nu}{N(N-1)} \Z^{\top}\rr_\z^c = GC_\nu(\X) = GC_\nu(\Z),
\]
which ends the proof.
\end{proof}

Finally, under a normality assumption, the generalized Gini covariance matrix $GC_\nu(\boldsymbol{X}) \equiv [GMD_\nu(X_{ k},X_{\ell})]$ is shown to be a positive semi-definite matrix.

\begin{theorem}\label{th2}
Let $Z \sim \mathcal{N}(0,1)$. If $\boldsymbol{X}$ represents identically distributed Gaussian random variables, with distribution $\mathcal{N}(\mu,\sigma^2)$, then the two following assertions holds:

\emph{(i)} $GC_\nu(\boldsymbol{X})= \sigma^{-1}\emph{Cov}(Z,\Phi(Z)) \emph{Var}(\boldsymbol{X}).$

\emph{(ii)} $GC_\nu(\boldsymbol{X})$ is a positive semi-definite matrix. 
\end{theorem}

\begin{proof}
The first part (i) follows from \cite{Yitzhaki13}, Chapter 6. The second part follows directly from (i). 
\end{proof}
% Théorème OK. Insérer ces résultats dans la PCA Gini...

Theorem \ref{th2} shows that under the normality assumption, the variance is a special case of the Gini methodology. As a consequence, for multivariate normal distributions, it is shown in Section \ref{s4} that Gini PCAs and classical PCA (based on the $\ell_2$ norm and the covariance matrix) are equivalent.

\section{Generalized Gini PCA}\label{s4}

In this section, the multidimensional Gini variability of the observations $i=1,\ldots,N$, embodied by the matrix $GC_\nu(\Z)$, is maximized in the $\R^{K}$-Euclidean space, \emph{i.e.}, in the set of variables $\{\z_{\cdot 1},\ldots,\z_{\cdot K}\}$. This allows the observations to be projected onto the new vector space spanned by the eigenvectors of $GC_\nu(\Z)$. Then, the projection of the variables is investigated in the $\R^{N}$-Euclidean space induced by $GC_\nu(\Z)$. Both observations and variables are analyzed through the prism of \emph{absolute} and \emph{relative} contributions to propose relevant interpretations of the data in each subspace.

\subsection{The $\R^{K}$-Euclidean space}

It is possible to investigate the projection of the data $\Z$ onto the new vector space induced by $GMD_\nu(\Z)$ or alternatively by $GC_\nu(\Z)$ since $GMD_\nu(\Z) = GC_\nu(\Z)$. Let $\f_{\cdot k}$ be the $k$th principal component, \emph{i.e.} the $k$th axis of the new subspace, such that the $N\times K$ matrix $\F$ is defined by $\F \equiv [\f_{\cdot 1},\ldots,\f_{\cdot K}]$ with $\rr_\f^c \equiv [\rv_{c,\f_1}^{\nu-1},\ldots,\rv_{c,\f_{\cdot K}}^c]$ its corresponding decumulative centered rank matrix (where each decumulative rank vector is raised to an exponent of $\nu-1$). The $K\times K$ matrix $\B \equiv [\bb_{\cdot 1},\ldots,\bb_{\cdot K}]$ is the projector of the observations, with the normalization condition $\bb_{\cdot k}^{\top}\bb_{\cdot k} = 1$, such that $\F=\Z\B$. We denote by $\lambda_{\cdot k}$ (or $2\mu_{\cdot k}$) the eigenvalues of the matrix $[GC_\nu(\Z)+GC_\nu(\Z)^{\top}]$. Let the basis $\mathscr{B}:=\{\bb_{\cdot1},\ldots,\bb_{\cdot h}\}$ with $h\leq K$ issued from the maximization of the overall Gini variability:
\[
\max \bb_{\cdot k}^{\top} GC_\nu(\Z) \bb_{\cdot k} \ \Longrightarrow \ [GC_\nu(\Z)+GC_\nu(\Z)^{\top}] \bb_{\cdot k} = 2\mu_{\cdot k} \bb_{\cdot k}, \ \forall k=1,\ldots,K.
\]
Indeed, from the Lagrangian,
\[
L = \bb_{\cdot k}^{\top} GC_\nu(\Z) \bb_{\cdot k} - \mu_{\cdot k} [1- \bb_{\cdot k}^{\top}\bb_{\cdot k}],
\]
because of the non-symmetry of $GC_\nu(\Z)$, the eigenvalue equation is,
\[
[GC_\nu(\Z) + GC_\nu(\Z)^{\top}]\bb_{\cdot k} = 2\mu_{\cdot k}\bb_{\cdot k},
\]
that is,
\begin{equation}\label{eigen-value-eq}
[GC_\nu(\Z)+GC_\nu(\Z)^{\top}]\bb_{\cdot k} = \lambda_{\cdot k}\bb_{\cdot k}.
\end{equation}
The new subspace $\{\f_{\cdot 1},\ldots,\f_{\cdot h}\}$ such that $h\leq K$ is issued from the maximization of the Gini variability between the observations on each axis $\f_{\cdot k}$. Although the result of the generalized Gini PCA seems to be close to the classical PCA, some differences exist.

\begin{proposition}
Let $\mathscr{B} = \{\bb_{\cdot 1},\ldots,\bb_{\cdot h}\}$ with $h\leq K$ be the basis issued from the maximization of $\bb_{\cdot k}^{\top} GC_\nu(\Z) \bb_{\cdot k}$ for all $k=1,\ldots,K$, then the following assertions hold:

\noindent\emph{(i)} $\max GMD_\nu(\f_{\cdot k},\f_{\cdot k}) = \mu_{\cdot k}$ for all $k=1,\ldots,K$, if and only if $\rv_{c,\f_{\cdot k}}^{\nu-1} = \rr_\z^c \bb_{\cdot k}$.

\noindent\emph{(ii)} $\bb_{\cdot k} \bb_{\cdot h}^{\top} = 0$, for all $k\neq h$.

\noindent\emph{(iii)} $\bb_{\cdot k} \bb_{\cdot k}^{\top} = 1$, for all $k=1,\ldots,K$.
\end{proposition}

\begin{proof}
(i) Note that $GMD_\nu(\f_{\cdot k},\f_{\cdot k}) = -\frac{2\nu}{N(N-1)}(\f_{\cdot k} - \bar{\f})^{\top}\rv_{c,\f_{\cdot k}}^{\nu-1}$, where $\rv_{c,\f_{\cdot k}}^{\nu-1}$ is the $k$th column of the centered (decumulative) rank matrix $\rr_\f^c$. Since $\f_{\cdot k} = \Z \bb_{\cdot k}$ and $\bar{\f} = (\bar{\f}_{\cdot 1},\ldots,\bar{\f}_{\cdot K}) = \mathbf{0}$ then:\footnote{We have:
\[
\bar{\mathbf{f}}_{\cdot k} = 1/N\sum_{i=1}^N f_{ik} = 1/N \left[ \sum_{i=1}^N \z^{\top}_{i\cdot} \mathbf{b}_{\cdot k} \right] = 1/N\left[ \sum_{i=1}^N \z^{\top}_{i\cdot} \right]\mathbf{b}_{\cdot k} = \mathbf{0}.\]} 
\begin{align}
\bb_{\cdot k}^{\top} GC_\nu(\Z) \bb_{\cdot k} &= -\frac{2\nu}{N(N-1)}\bb_{\cdot k}^{\top} \Z^{\top} \rr_{\z}^c \bb_{\cdot k} \\
&= -\frac{2\nu}{N(N-1)}\bb_{\cdot k}^{\top} \Z^{\top} \rv_{c,\f_{\cdot k}}^{\nu-1} \ \ \ \ (\text{by} \ \rv_{c,\f_{\cdot k}}^{\nu-1} = \rr_\z^c \bb_{\cdot k}) \notag \\
&= -\frac{2\nu}{N(N-1)}\f^{\top}_{\cdot k} \rv_{c,\f_{\cdot k}}^{\nu-1} \notag \\
&= GMD_\nu(\f_{\cdot k},\f_{\cdot k}).\notag 
\end{align} 
Then, maximizing the multidimensional variability $\bb_{\cdot k}^{\top} GC_\nu(\Z) \bb_{\cdot k}$ yields from (\ref{eigen-value-eq}): 
\begin{align}
\bb_{\cdot k}^{\top} [GC_\nu(\Z)+GC_\nu(\Z)^{\top}]\bb_{\cdot k} &= \bb_{\cdot k}^{\top} \lambda_{\cdot k}\bb_{\cdot k}\notag \\
\Longleftrightarrow \ \bb_{\cdot k}^{\top} GC_\nu(\Z)\bb_{\cdot k} + \bb_{\cdot k}^{\top}GC_\nu(\Z)^{\top} \bb_{\cdot k} &= \lambda_{\cdot k}.\notag
\end{align}
Since $\bb_{\cdot k}^{\top} GC_\nu(\Z) \bb_{\cdot k} = \bb_{\cdot k}^{\top} GC_\nu(\Z)^{\top} \bb_{\cdot k}$, then 
\[
\bb_{\cdot k}^{\top}GC_\nu(\Z)^{\top} \bb_{\cdot k} = \lambda_{\cdot k}/2 = \mu_{\cdot k},
\]
and so $GMD_\nu(\f_{\cdot k},\f_{\cdot k}) = \lambda_{\cdot k}$ for all $k=1,\ldots,K$. The results (ii) and (iii) are straightforward.
\end{proof}

\subsection{Discussion}

Condition (i) shows that the maximization of the multidimensional variability (in the Gini sense) $\bb_{\cdot k}^{\top} GC_\nu(\Z) \bb_{\cdot k}$ does not necessarily coincide with the maximization of the variability of the observations projected onto the new axis $\f_{\cdot k}$ embodied by $GMD_\nu(\f_{\cdot k},\f_{\cdot k})$. Since in general, the rank of the observations on axis $\f_{\cdot k}$ does not coincide with the projected ranks, that is,
\[
\rv_{c,\f_{\cdot k}}^{\nu-1} \neq \rr_\z^c \bb_{\cdot k},
\]
then,
\[
\max \bb_{\cdot k}^{\top} GC(\Z) \bb_{\cdot k} \neq GMD_\nu(\f_{\cdot k},\f_{\cdot k}).
\]
In other words, maximizing the quadratic form $\bb_{\cdot k}^{\top} GC(\Z)_\nu \bb_{\cdot k}$ does not systematically maximize the overall Gini variability $GMD_\nu(\f_{\cdot k},\f_{\cdot k})$. However, it maximizes the following generalized Gini index:
\begin{align}
GGMD_\nu(\f_{\cdot k},\f_{\cdot k}) &:= -\frac{2\nu}{N(N-1)}\bb_{\cdot k}^{\top} \Z^{\top} \rr_{\z}^c \bb_{\cdot k} \notag \\
& = -\frac{2\nu}{N(N-1)} \f_{\cdot k}^{\top}\bb_{\cdot k}^{\top} (\rr_{\z}^{c})^{\top}. \notag
\end{align} 
In the literature on inequality indices, this kind of index is rather known as a generalized Gini index, because of the product between a variable $\f_{\cdot k}$ and a function $\Psi$ of its ranks, $\Psi (\rv_{\f_{k}}) := \bb_{\cdot k}^{\top} (\rr_{\z}^{c})^{\top}$, such that:
\[
GGMD_\nu(\f_{\cdot k},\f_{\cdot k}) = -\frac{2\nu}{N(N-1)}\f_{\cdot k} \ \Psi (\rv_{\f_{k}}).
\]
\cite{Yaari1987} and subsequently \cite{Yaari1988} proposes generalized Gini indices with a rank distortion function $\Psi$ that describes the behavior of the decision maker (being either max-min or max-max).\footnote{Strictly speaking \cite{Yaari1987} and \cite{Yaari1988} suggests probability distortion functions $\Psi: [0,1] \rightarrow [0,1]$, which does not necessarily coincide to our case.} 

It is noteworthy that this generalized Gini index of variability is very different from \cite{Banerjee}'s multidimensional Gini index. The author proposes to extract the first eigenvector $\e_{\cdot 1}$ of $\X^{\top}\X$ and to project the data $\X$ such that $\s:=\X\e_{\cdot 1}$ so that the multidimensional Gini index is $G(\s) = \s^{\top} \tilde{\Psi}(\rv_s)$, with $\rv_s$ the rank vector of $\s$ and with $\tilde{\Psi}$ a function that distorts the ranks. \cite{Banerjee}'s index is derived from the matrix $\X^{\top}\X$. To be precise, the maximization of the variance-covariance matrix $\X^{\top}\X$ (based on the $\ell_2$ metric) yields the projection of the data on the first component $\f_{\cdot 1}$, which is then employed in the multidimensional Gini index (based on the $\ell_1$ metric). This approach is legitimated by the fact that $G(\s)$ has some desirable properties linked with the Gini index. However, this Gini index deals with an information issued from the variance, because the vector $\s$ relies on the maximization of the variance of component $\f_{\cdot 1}$. Alternatively, it is possible to make use of the Gini variability, in a first stage, in order to project the data onto a new subspace, and in a second stage, to use the generalized Gini index of the projected data for the interpretations. In such as case, the Gini metric enables outliers to be attenuated. The employ of $G(\s)$ as a result of the variance-covariance maximization may transform the data so that outlying observations would capture an important part of the information (variance) on the first component. This case occurs in the classical PCA. This fact will be proven in the next sections with Monte Carlo simulations. Let us before investigate the employ of the generalized Gini index ${GGMD}_{\nu}$.

\subsection{Properties of ${GGMD}_{\nu}$}

Since the Gini PCA relies on the generalized Gini index ${GGMD}_{\nu}$, let us explore its properties. 

\begin{proposition}\label{prop4.2}
Let the eigenvalues of $GC_{\nu}(\Z)+GC_{\nu}(\Z)^{\top}$ be such that $\lambda_1 = \mu_1/2 \geq \cdots \geq \lambda_{\cdot K} = \mu_{\cdot K}/2$. Then, 

\noindent \emph{(i)} $GGMD_\nu(\f_{\cdot k},\f_{\cdot k}) = GGMD_\nu(\f_{\cdot k},\f_{\cdot \ell}) = GGMD_\nu(\f_{\cdot \ell},\f_{\cdot k}) = 0, \ \text{for all } \ell=1,\ldots,K$, if and only if, $\lambda_{\cdot k} =0$.

\noindent \emph{(ii)} $\max_{k=1,\ldots,K} {GGMD}_{\nu}(\f_{\cdot k},\f_{\cdot k}) = \mu_1$.

\noindent \emph{(iii)} $\min_{k=1,\ldots,K} {GGMD}_{\nu}(\f_{\cdot k},\f_{\cdot k}) = \mu_{\cdot K}$.
 \end{proposition}

\begin{proof}
(i) The result comes from the rank-nullity theorem. From the eigenvalue Equation (\ref{eigen-value-eq}), we have:
\[
\bb_{\cdot k}^{\top}GC_{\nu}(\Z) \bb_{\cdot k} = \lambda_{\cdot k}/2 = \mu_{\cdot k}.
\]
Let $f$ be the linear application issued from the matrix $GC_{\nu}(\Z)$. Whenever $\lambda_{\cdot k} = 0$, two columns (or rows) of $GC_{\nu}(\Z)$ are collinear, then the dimension of the image set of $f$ is $\dim(f) = K-1$. Hence, $\f_{\cdot k}=\mathbf{0}$. Since $\bb_{\cdot k}^{\top}GC_{\nu}(\Z)^{\top} \bb_{\cdot k} = {GGMD}_\nu(\f_{\cdot k},\f_{\cdot k})$ for all $k=1,\ldots,K$, then for $\lambda_{\cdot k}$ we get:
\[
\bb_{\cdot k}^{\top}GC_{\nu}(\Z)^{\top} \bb_{\cdot k} = {GGMD}_\nu(\f_{\cdot k},\f_{\cdot k}) =\lambda_{\cdot k}/2 = \mu_{\cdot k} = 0.
\]
On the other hand, since $\f_{\cdot k} = \mathbf{0}$, it follows that $GGMD_\nu(\f_{\cdot k},\f_{\cdot \ell}) = 0$ for all $\ell=1,\ldots,K$. Also, if $\f_{\cdot k} = \mathbf{0}$ then the centered rank vector $\rv^c_{\f_{\cdot k}} = \mathbf{0}$, and so $GGMD_\nu(\f_{\cdot \ell},\f_{\cdot k}) = 0$ for all $\ell=1,\ldots,K$.

\noindent (ii) The proof comes from the Rayleigh-Ritz identity:
\[
\lambda_{\max} := \max \frac{\bb_{\cdot 1}^{\top}[GC_\nu(\Z)+GC_\nu(\Z)^{\top}]\bb_1}{\bb_1^{\top}\bb_1} = \lambda_{1}.
\]
Since $\bb_1^{\top}GC_\nu(\Z)\bb_{\cdot 1} = \lambda_1/2$ and because $\bb_1^{\top}GC_\nu(\Z)\bb_1 = GGMD_{\nu}(\f_1,\f_1)$, the result follows. 

\noindent (iii) Again, the Rayleigh-Ritz identity yields:
\[
\lambda_{\min} := \min \frac{\bb_{\cdot K}^{\top}[GC_\nu(\Z)+GC_\nu(\Z)^{\top}]\bb_{\cdot K}}{\bb_{\cdot K}^{\top}\bb_{\cdot K}} = \lambda_{K}.
\]
Then, $\bb_{\cdot K}^{\top}GC_\nu(\Z)\bb_{\cdot K} = GGMD_{\nu}(\f_{\cdot K},\f_{\cdot K}) = \lambda_{\cdot K}/2$.
\end{proof}

The index $GGMD_{\nu}(\f_{\cdot k},\f_{\cdot k})$ represents the variability of the observations projected onto component $\f_{\cdot k}$. When this variability is null, then the eigenvalue is null (i). In the same time, there is neither co-variability in the Gini sense between $\f_{\cdot k}$ and another axis $\f_{\cdot \ell}$, that is $GGMD_{\nu}(\f_{\cdot k},\f_{\cdot \ell})=0$. 

In the Gaussian case, because the Gini correlation matrix is positive semi-definite, the eigenvalues are non-negative, then $GGMD$ is null whenever it reaches its minimum.

\begin{proposition}
Let $Z \sim \mathcal N (0,1)$ and let $\boldsymbol{X}$ represent identically distributed Gaussian random variables, with distribution $\mathcal{N}(\mathbf{0},\boldsymbol{\rho})$ such that $\emph{Var}(X_k)=1$ for all $k=1,\ldots,K$ and let $\gamma_1,\ldots,\gamma_K$ be the eigenvalues of $\emph{Var}(\boldsymbol{X})$. Then the following assertions holds:

\emph{(i)} $\emph{Tr}[GC_\nu(\boldsymbol{X})] = \emph{Cov}(Z,\Phi(Z))\emph{Tr}[\emph{Var}(\boldsymbol{X})]$. 

\emph{(ii)} $\mu_k = \emph{Cov}(Z,\Phi(Z)) \gamma_k$ for all $k=1,\ldots,K$.

\emph{(iii)} $|GC_\nu(\boldsymbol{X})| = \emph{Cov}^K(Z,\Phi(Z))|\emph{Var}(\boldsymbol{X})|$.

\emph{(iv)} For all $\nu > 1$:
$$
\frac{\mu_k}{\emph{Tr}[GC_\nu(\X)]} = \frac{\gamma_k}{\emph{Tr}[\emph{Var}(\X)]}, \ \forall k=1,\ldots,K.
$$
\end{proposition}

\begin{proof}
From Theorem \ref{th2}: 
$$GC_\nu(\boldsymbol{X})=\sigma^{-1} \text{Cov}(Z,\Phi(Z)) \emph{Var}(\boldsymbol{X}).$$
From \cite{Abramowitz} (Chapter 26), when $Z \sim \mathcal{N}(0,1)$, 
$$
\text{Cov}(Z,\Phi(Z))=\frac{1}{2\sqrt{\pi}}\approx 0.2821.
$$
Then the results follow directly. 
\end{proof}

Point (iv) shows that the eigenvalues of the standard PCA are proportional to those issued from the generalized Gini PCA. Because each eigenvalue (in proportion of the trace) represents the variability (or the quantity of information) inherent to each axis, then both PCA techniques are equivalent when $\X$ is Gaussian:
$$
\frac{\mu_k}{\text{Tr}[GC_\nu(\X)]} = \frac{\gamma_k}{\text{Tr}[\text{Var}(\X)]}, \ \forall k=1,\ldots,K \ ; \ \forall \nu > 1.
$$

\subsection{The $\R^{N}$-Euclidean space}

In classical PCA, the duality between $\R^{N}$ and $\R^{K}$ enables the eigenvectors and eigenvalues of $\R^{N}$ to be deduced from those of $\R^{K}$ and conversely. This duality is not so obvious in the Gini PCA case. Indeed, in $\R^{N}$ the Gini variability between the observations would be measured by $GC_\nu(\widetilde{\Z}):=\frac{-2\nu}{N(N-1)}(\rr_\z^{c})^{\top}\Z$, and subsequently the idea would be to derive the eigenvalue equation related to $\R^N$,
\[
[GC_\nu(\widetilde{\Z})+GC_\nu(\widetilde{\Z})^{\top}]\ \widetilde{\bb}_{\cdot k} = \widetilde{\lambda}_{\cdot k}\widetilde{\bb}_{\cdot k}.
\]   
The other option is to define a basis of $\R^N$ from a basis already available in $\R^K$. In particular, the set of principal components $\{\f_1,\ldots,\f_{\cdot k}\}$ provides by construction a set of normalized and orthogonal vectors. Let us rescale the vectors $\f_{\cdot k}$ such that:
\[
\widetilde{\f}_{\cdot k} = \frac{\f_{\cdot k}}{GMD_\nu(\f_{\cdot k},\f_{\cdot k})}.
\]
Then, $\{\widetilde{\f}_1,\ldots,\widetilde{\f}_{\cdot k}\}$ constitutes an orthonormal basis of $\R^K$ in the Gini sense since $GMD_\nu(\widetilde{\f}_{\cdot k},\widetilde{\f}_{\cdot k})=1$. This basis may be used as a projector of the variables $\z_{\cdot k}$ onto $\R^N$. Let $\widetilde{\F}$ be the $N\times K$ matrix with $\widetilde{\f}_{\cdot k}$ in columns. The projection of the variables $\z_{\cdot k}$ in $\R^N$ is given by the following Gini correlation matrix:
\[
\V := \frac{-2\nu}{N(N-1)}\ \widetilde{\F}^{\top}\rr_\z^c,
\]
whereas it is given by $\frac{1}{N}\widetilde{\F}^{\top}\Z$ in the standard PCA, that is, the matrix of Pearson correlation coefficients between all $\widetilde{\f}_{\cdot k}$ and $\z_{\cdot \ell}$. The same interpretation is available in the Gini case. The matrix $\V$ is normalized in such a way that $\V \equiv [v_{k\ell}]$ are the $G$-correlations indices between $\widetilde{\f}_{\cdot k}$ and $\z_{\cdot \ell}$. This yields the ability to make easier the interpretation of the variables projected onto the new subspace.

\section{Interpretations of the Gini PCA}\label{s5}

The analysis of the projections of the observations and of the variables are necessary to provide accurate interpretations. Some criteria have to be designed in order to bring out, in the new subspace, the most significant observations and variables. 

\subsection{Observations}
 
The absolute contribution of an observation $i$ to the variability of a principal component $\f_{\cdot k}$ is:
\[
ACT_{ik} = \frac{f_{ik}\ \Psi(R(f_{ik}))}{GGMD_\nu(\f_{\cdot k},\f_{\cdot k})}.
\]
The absolute contribution of each observation $i$ to the generalized Gini Mean Difference of $\f_{\cdot k}$ ($ACT_{ik}$) is interpreted as a percentage of variability of $GGMD_\nu(\f_{\cdot k},\f_{\cdot k})$, such that $\sum_{i=1}^N ACT_{ik}=1$. This provides the most important observations $i$ related to component $\f_{\cdot k}$ with respect to the information $GGMD_\nu(\f_{\cdot k},\f_{\cdot k})$. On the other hand, instead of employing the Euclidean distance between one observation $i$ and the component $\f_{\cdot k}$, the Manhattan distance is used. The relative contribution of an observation $i$ to component $\f_{\cdot k}$ is then:
\[
RCT_{ik} = \frac{\left|f_{ik}\right |}{ \| \f_{i\cdot}\|_1}.
\]
Remark that the gravity center of $\{\f_1,\ldots,\f_{\cdot K}\}$ is $\mathbf{g}:=(\bar{\f}_1,\ldots,\bar{\f}_{\cdot K}) =\mathbf{0}$. The Manhattan distance between observation $i$ and $\mathbf{g}$ is then $\sum_{k=1}^K\left|f_{ik} - 0\right |$, and so 
\[
RCT_{ik} = \frac{\left|f_{ik}\right |}{\| \f_{i\cdot} - \mathbf g\|_1}.
\]
The relative contribution $RCT_{ik}$ may be interpreted rather as the contribution of dimension $k$ to the overall distance between observation $i$ and $\mathbf{g}$.

\subsection{Variables}

The most significant variables must be retained for the analysis and the interpretation of the data in the new subspace. It would be possible, in the same manner as in the observations case, to compute absolute and relative contributions from the Gini correlation matrix $\V \equiv [v_{k\ell}]$. Instead, it is possible to test directly for the significance of the elements $v_{k\ell}$ of $\V$ in order to capture the variables that significantly contribute to the Gini variability of components $\f_{\cdot k}$. Let us denote $\widetilde{U}_{\ell k}:=\cov(\f_{\cdot \ell},\rr_{\z_{\cdot k}}^c)$ with $\rr_{\z_{\cdot k}}^c$ the (decumulative) centered rank vector of $\z_{\cdot k}$ raised to an exponent of $\nu-1$ and $U_{\cdot \ell} :=\cov(\f_{\cdot \ell},\rr^c_{\f_{\cdot \ell}})$. Those two Gini covariances yield the following $U$-statistics:     
\begin{equation} \notag
U_{\ell k} = \frac{\widetilde{U}_{\ell k}}{U_{\cdot \ell}} = v_{k\ell}.
\end{equation}
Let $U_{\ell k}^0$ be the expectation of $U_{\ell k}$, that is $U_{\ell k}^0:=\E[U_{\ell k}]$. From \cite{Yitzhaki13}, $U_{\ell k}$ is an unbiased and consistent estimator of $U_{\ell k}^0$. From Theorem 10.4 in \cite{Yitzhaki13}, Chapter 10, we asymptotically get that $\sqrt{N}(U_{\ell k} - U_{\ell k}^0) \stackrel{a}{\sim} \mathcal{N}$. Then, it is possible to test for:
\begin{equation}\notag
\left\|
\begin{array}{l}
H_0 : U_{\ell k}^0 = 0  \\ 
H_1 : U_{\ell k}^0 \neq 0.
\end{array}
\right.
\end{equation}
Let $\hat{\sigma}^2_{\ell k}$ the Jackknife variance of $U_{\ell k}$, then it is possible to test for the null under the assumption $N \to \infty$ as follows:\footnote{As indicated by \cite{Yitzhaki1990}, the efficient Jackknife method may be used to find the variance of any $U$-statistics.}
\begin{equation} \notag
\frac{U_{\ell k}}{\hat{\sigma}_{\ell k}} \stackrel{a}{\sim} \mathcal{N}(0,1).
\end{equation}
The usual PCA enables the variables to be analyzed in the circle of correlation, which outlines the correlations between the variables $\z_{\cdot k}$ and the components $\f_{\cdot \ell}$. In order to make a comparison with the usual PCA, let us rescale the $U$-statistics $U_{\ell k}$. Let $\mathbf{U}$ be the $K\times K$ matrix such that $\mathbf{U} \equiv [U_{\ell k}]$, and $\mathbf{u}_{\cdot k}$ the $k$-th column of $\mathbf{U}$. Then, the absolute contribution of the variable $\z_{\cdot k}$ to the component $\f_{\cdot \ell}$ is:
\[
\widetilde{ACT}_{k \ell} = \frac{U_{\ell k}}{\| \mathbf{u}_{\cdot k}\|_2}.
\]
The measure $\widetilde{ACT}_{k \ell}$ yields a graphical tool aiming at comparing the standard PCA with the Gini PCA. In the standard PCA, $\cos^2\theta$ (see Figure \ref{circle} below) provides the Pearson correlation coefficient between $\f_1$ and $\z_{\cdot k}$. In the Gini PCA, $\cos^2\theta$ is the normalized Gini correlation coefficient $\widetilde{ACT}_{k 1}$ thanks to the $\ell_2$ norm. 

\begin{figure}[H] 
\centering\includegraphics[scale=0.7]{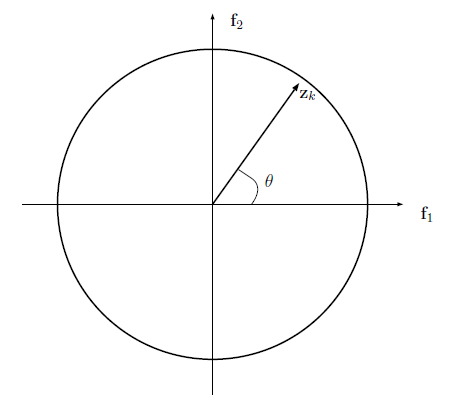} 
\caption{Circle of correlation}\label{circle}
\end{figure}

It is worth mentioning that the circle of correlation does not provide the significance of the variables. This significance relies on the statistical test based on the $U$-statistics exposed before. Because $\widetilde{ACT}$ depends on the $\ell_2$ metric, it is sensitive to outliers, and as such, the choice of the variables must rely on the test of $U_{\ell k}^0$ only. 

%Finally, the circle of correlation allows to check graphically whether some components are necessary or not. For instance, if vector $\z_{\cdot k}$ is close to the circle (as in Figure 1 above), the vector is well defined in the plan $(\f_1,\f_2)$, it is then not necessary to look for some correlation between $z_{\cdot k}$ and another $\f_{\cdot \ell}$. 

\section{Monte Carlo Simulations}\label{s6}

In this Section, it is shown with the aid of Monte Carlo simulations that the usual PCA yields irrelevant results when outlying observations contaminate the data. To be precise, the absolute contributions computed in the standard PCA based on the variance may lead to select outlying observations on the first component in which there is the most important variability (a direct implication of the maximization of the variance). In consequence, the interpretation of the PCA may inflate the role of the first principal components. The Gini PCA dilutes the importance of the outliers to make the interpretations more robust and therefore more relevant.    

\medskip
 
\begin{algorithm}[H]
\KwResult{Robust Gini PCA with data contamination}
 $\theta=1$ $[$ $\theta$ is the value of the outlier$]$  \;
 \Repeat{ $\theta = 1000$ \emph{[increment of $1$]}}{
Generate a $4$-variate normal distribution $\X \sim \mathcal{N}$, $N=500$ \;
Introduce outliers in 1 row of $\X$: $\X^o_{ji}: = \theta \X_{ji}$ with $j=1,\ldots,4$ [for a random row localization]\;
For each method (Variance and Gini), the $ACT$ and $RCT$ are computed for the axes 1 and 2 on the contaminated matrix $\X^o$\;
}
\Return Mean squared Errors of eigenvalues, $ACT$ and $RCT$ \;
\caption{Monte Carlo Simulation}\label{MC}
\end{algorithm}
\bigskip
The mean squared errors of the eigenvalues are computed as follows:
\begin{equation}\notag
MSE_{\lambda_k} = \frac{\sum_{i=1}^{1,000}(\lambda_k^{oi} - \lambda_k)^2}{1,000},
\end{equation} 
where $\lambda_k^{oi}$ is the eigenvalue computed with outlying observations in the sample. The MSE of $ACT$ et $RCT$ are computed in the same manner.

\medskip

We first investigate the case where the variables are highly correlated in order to gauge the robustness of each technique (Gini for $\nu=2,4,6$ and variance). The correlation matrix between the variables is given by:

\begin{equation}\notag
\boldsymbol{\rho} = \left( \begin{array}{cccc}
1 & 0.8 & 0.9 & 0.7 \\
0.8 & 1 & 0.8 & 0.75 \\
0.9 & 0.8 & 1 & 0.6  \\
0.7 & 0.75 & 0.6 & 1  \\
 \end{array}\right)
\end{equation}

\bigskip

As can be seen in the matrix above, we can expect that all the information be gathered on the first axis because each pair of variables records an important linear correlation. The repartition of the information on each component, that is, each eigenvalue in percentage of  the sum of the eigenvalues is the following. 

\begin{table}[h!]
\begin{tabular}{| c | c || c | c | c | c |}
\hline \multicolumn{2}{|c||} {Eigenvalues}  & Gini $\nu$ = 2 &  Gini $\nu$ = 4 & Gini $\nu$ = 6 & Variance  \\
\hline
\hline  
            & eigenvalues (\%) & 81.65341 & 82.31098  & 82.28372 & 81.11458 \\
 \cline{3-6}   Axis 1 & &   &  &   &  \\
           & \textbf{
					 MSE} & \textbf{ 12.313750} & \textbf{12.196975} & \textbf{12.221840} & \cellcolor{lightgray} \textbf{15.972710}  \\
\hline 
\hline
           & eigenvalues (\%) & 10.90079 & 10.47317 & 10.46846 & 11.35471 \\
 \cline{3-6} Axis 2 & &   &  &   &  \\
          & \textbf{MSE} & \cellcolor{lightgray} \textbf{ 11.478541}  & \textbf{11.204504} & \textbf{10.818344}  &  \textbf{10.688924} \\
\hline 
\hline 
           & eigenvalues (\%) & 5.062329 & 4.996538  & 5.088865 & 5.112817 \\
\cline{3-6} Axis 3 &  &   &  &   &  \\
           & \textbf{MSE} & \textbf{ 2.605312} &  \textbf{2.608647} & \textbf{2.799180} & \cellcolor{lightgray} \textbf{4.687323} \\
\hline
\hline 
            & eigenvalues (\%) & 2.383476 & 2.219311 & 2.15896 & 2.417897 \\
\cline{3-6} Axis 4 &   &   &  &   &  \\
           & \textbf{MSE} & \textbf{ 1.541453}  & \textbf{1.826055} & \cellcolor{lightgray} \textbf{3.068596} & \textbf{2.295100} \\
\hline 
\end{tabular}
\caption{Eigenvalues and their MSE} 
 \end{table}

The first axis captures around 82\% of the variability of the overall sample (before contamination). Although each PCA method yields the same repartition of the information over the different components before the contamination of the data, it is possible to show that the classical PCA is not robust. For this purpose, let us analyze Figures \ref{fig:3a}-\ref{fig:3d} below that depict the MSE of each observation with respect to the contamination process described in Algorithm 1 above. 

\begin{figure}[h!]
\begin{subfigure}{.6\textwidth}
\includegraphics[width=1\linewidth]{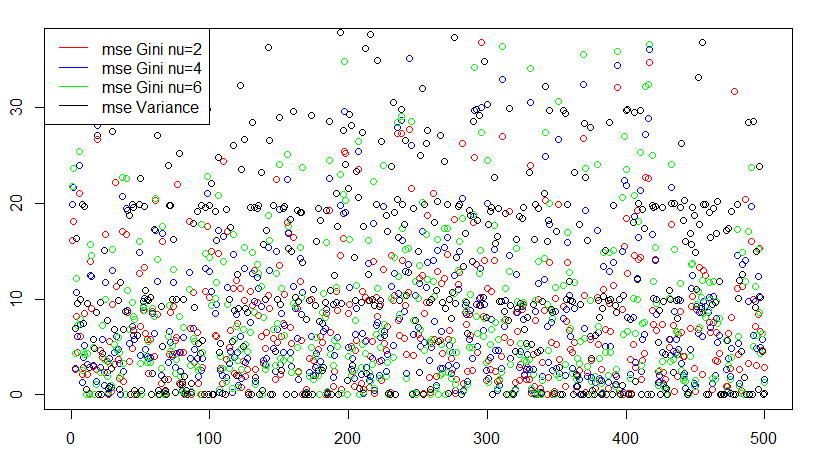} 
\caption{$ACT_1$}\label{fig:3a}
\end{subfigure}%
\begin{subfigure}{.6\textwidth}
\includegraphics[width=1\linewidth]{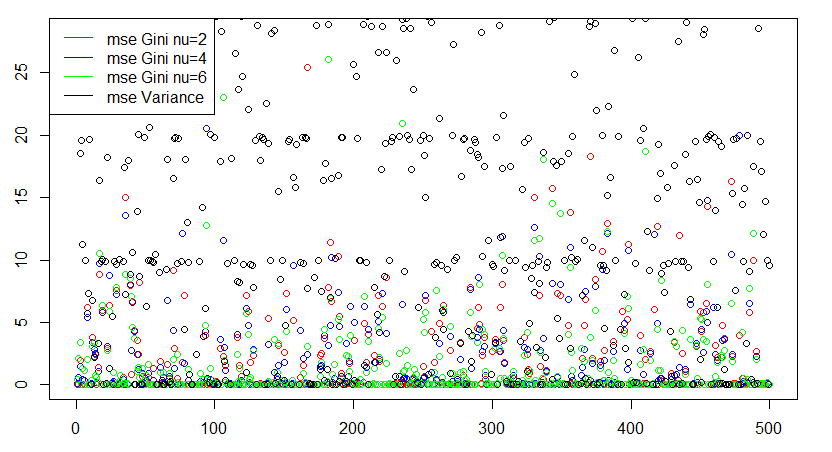} 
\caption{$ACT_2$}\label{fig:3b}
%\textbf{Figure 3a:} $ACT_1$ &  \textbf{Figure 3b:} $ACT_2$ \\
\end{subfigure}%

\begin{subfigure}{.6\textwidth}
\includegraphics[width=1\linewidth]{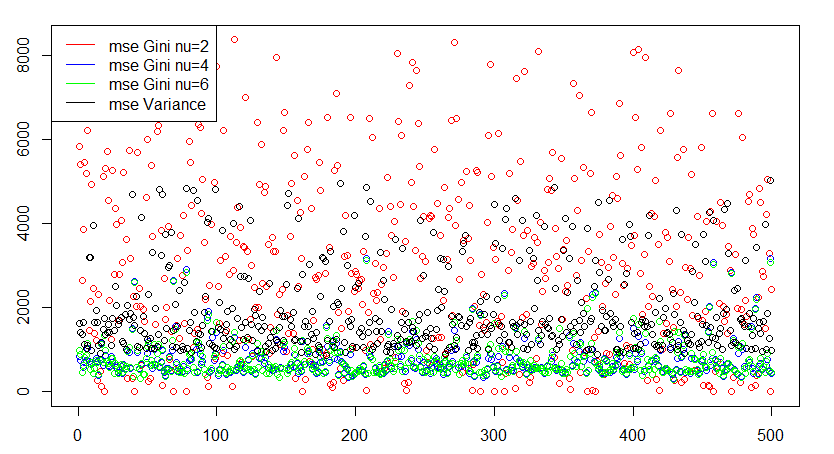} 
\caption{$RCT_1$}\label{fig:3c}
\end{subfigure}%
\begin{subfigure}{.6\textwidth}
\includegraphics[width=1\linewidth]{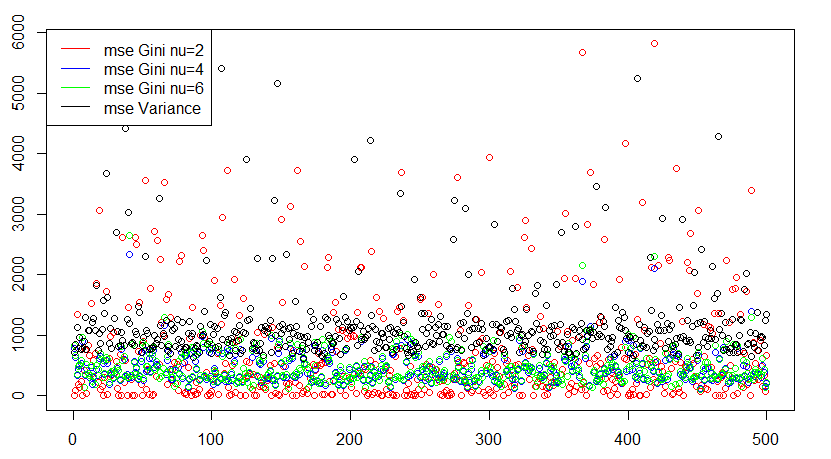} 
\caption{$RCT_2$}\label{fig:3d}
%\textbf{Figure 3c:} $RCT_1$ &  \textbf{Figure 3d:} $RCT_2$ 
\end{subfigure}%
\caption{$ACT_1$, $ACT_2$, $RCT_1$ and $RCT_2$}
\end{figure}

On the first axis of Figure \ref{fig:3a}, the absolute contribution of each observation (among 500 observations) is not stable because of the contamination of the data, however the Gini PCA performs better. The MSE of the ACTs measured during to the contamination process provides lower values for the Gini index compared with the variance. On the other hand, if we compute the standard deviation of all these MSEs over the two first axis, again the Gini methodology provides lower variations (see Table \ref{tab:2}). \break

\begin{table}[h!]
\begin{center}
\begin{tabular}{|c||cccc|}
\hline
&Gini $\nu=2$ &  Gini $\nu=4$ &  Gini $\nu=6$ & Variance \\ \hline\hline
Axis 1&6.08 & 6.62 & 7.41 & \textbf{12.09} \\
Axis 2&4.07 & 5.12 & 13.37 & 2.98 \\\hline 
\end{tabular}
\caption{Standard deviation of the MSE of the ACTs on the two first axis}\label{tab:2}
\end{center}
\end{table}

Let us take now an example with less correlations between the variables in order to get a more equal repartition of the information on the first two axes. 

\begin{equation}\notag
\mathbf{\rho} = \left( \begin{array}{cccc}
1 & -0.5 & 0.25 & 0.5 \\
-0.5 & 1 & -0.9 & 1 \\
0.25 & -0.9 & 1 & -0.25  \\
0.5 & 0 & -0.25 & 1  \\
 \end{array}\right)
\end{equation}

The repartition of the information over the new axes (percentage of each eigenvalue) is given in Table \ref{Tab3}. When the information is less concentrated on the first axis (55\% on axis 1 and around 35\% on axis 2), the MSE of the eigenvalues after contamination are much more important for the standard PCA compared with the Gini approach (2 to 3 times more important). Although the fourth axis reports an important MSE for the Gini method ($\nu=6$), the eigenvalue percentage is not significant (1.56\%).   

\begin{table}[h!]
\begin{tabular}{| c | c || c | c | c | c |}
\hline \multicolumn{2}{|c||} {eigenvalues (\%)}  & Gini $\nu$ = 2 &  Gini $\nu$ = 4 & Gini $\nu$ = 6 & Variance  \\
\hline
\hline  
            & eigenvalues (\%) & 55.3774 & 55.15931  & 54.96172 & 55.08917 \\
 \cline{3-6}   Axis 1 & &   &  &   &  \\
           & \textbf{MSE} & \textbf{ 17.711023} & \textbf{14.968196} & \textbf{12.745760} & \cellcolor{lightgray} \textbf{38.929147}  \\
\hline 
\hline
           & eigenvalues (\%) & 35.8385 & 35.86216 & 35.8745 & 36.06118 \\
 \cline{3-6} Axis 2 & &   &  &   &  \\
          & \textbf{MSE} & \textbf{ 14.012198}  & \textbf{16.330350} & \textbf{18.929923}  & \cellcolor{lightgray} \textbf{30.948674} \\
\hline 
\hline 
           & eigenvalues (\%) & 7.227274 & 7.345319  & 7.527222 & 7.329535 \\
\cline{3-6} Axis 3 &  &   &  &   &  \\
           & \textbf{MSE} & \textbf{ 4.919686} &  \textbf{4.897820} & \textbf{5.036241} & \cellcolor{lightgray} \textbf{6.814252} \\
\hline
\hline 
            & eigenvalues (\%) & 1.556831 & 1.633214 & 1.636561 & 1.520114 \\
\cline{3-6} Axis 4 &   &   &  &   &  \\
           & \textbf{MSE} & \textbf{ 1.149770}  & \textbf{7.890184} & \cellcolor{lightgray} \textbf{14.047539} & \textbf{1.438904} \\
\hline 
\end{tabular}
\caption{Eigenvalues and their MSE}\label{Tab3}
\end{table}

Let us now have a look on the MSE of the absolute contributions of each observation ($N=500$) for each PCA technique (\ref{fig:5a}-\ref{fig:5b}). We obtain the same kind of results, with less variability on the second axis. In Figures \ref{fig:5a}-\ref{fig:5b}, it is apparent that the classical PCA based on the $\ell_2$ norm exhibits much more ACT variability (black points). This means that the contamination of the data can lead to the interpretation of some observations as significant (important contribution to the variance of the axis) while they are not
(and vice versa). On the other hand, the MSE of the RCTs after contamination of the data, Figures \ref{fig:5c}-\ref{fig:5d}, are less spread out for the Gini technique for $\nu=4$ and $\nu=6$, however for $\nu=2$ there is more variability of the MSE compared with the variance. This means that the distance from one observation to an axis may not be reliable (although the interpretation of the data rather depends on the ACTs).

\begin{figure}[h!]
\begin{subfigure}{.6\textwidth}
\includegraphics[width=1\linewidth]{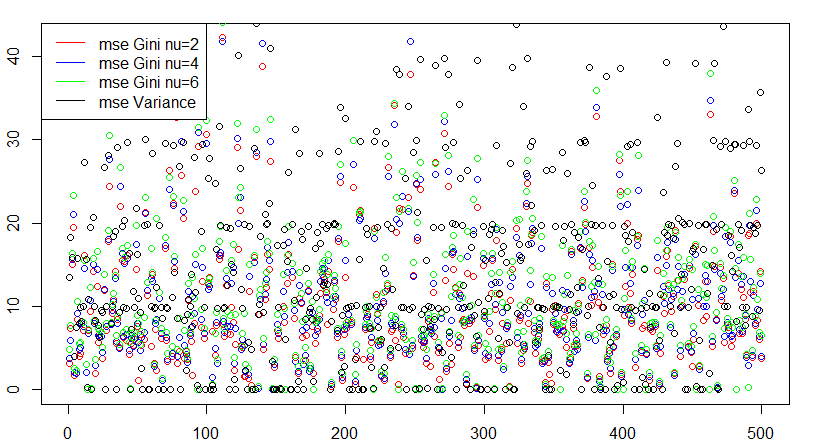} 
\caption{$ACT_1$}\label{fig:5a}
\end{subfigure}%
\begin{subfigure}{.6\textwidth}
\includegraphics[width=.98\linewidth]{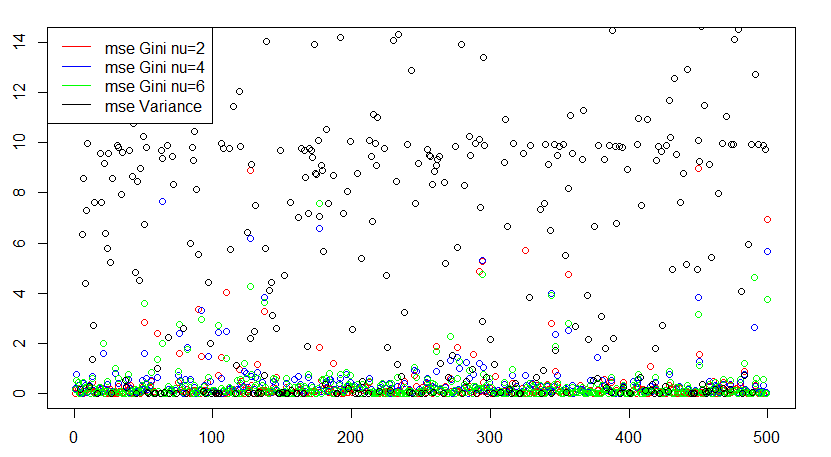} 
\caption{$ACT_2$}\label{fig:5b}
%\textbf{Figure 3a:} $ACT_1$ &  \textbf{Figure 3b:} $ACT_2$ \\
\end{subfigure}%

\begin{subfigure}{.6\textwidth}
\includegraphics[width=1\linewidth]{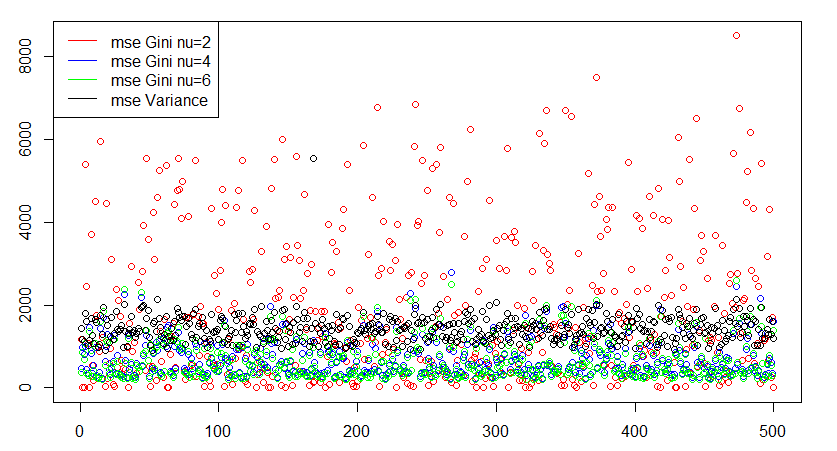} 
\caption{$RCT_1$}\label{fig:5c}
\end{subfigure}%
\begin{subfigure}{.6\textwidth}
\includegraphics[width=1\linewidth]{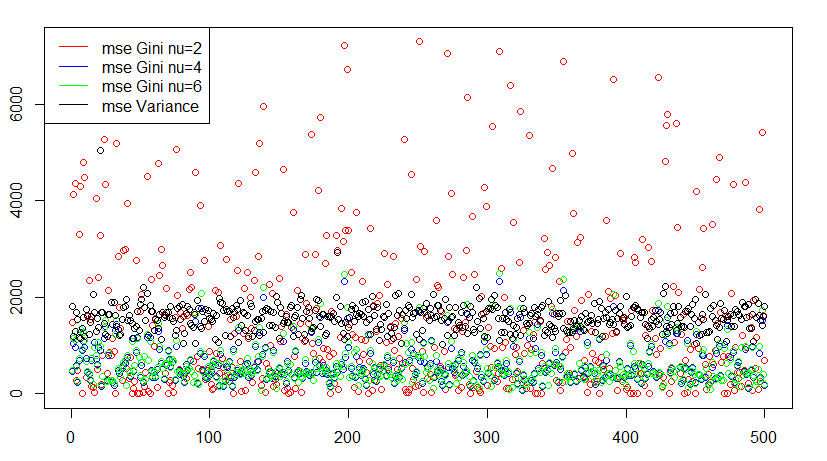} 
\caption{$RCT_2$}\label{fig:5d}
%\textbf{Figure 3c:} $RCT_1$ &  \textbf{Figure 3d:} $RCT_2$ 
\end{subfigure}%
\caption{$ACT_1$, $ACT_2$, $RCT_1$ and $RCT_2$}
\end{figure}

The results of Figures \ref{fig:5a}-\ref{fig:5d} can be synthesized by measuring the standard deviation of the MSE over the ACTs of the 500 observations along the two first axes. 

\begin{table}[h!]
\begin{center}
\begin{tabular}{|c||cccc|}
\hline
&Gini $\nu=2$ &  Gini $\nu=4$ &  Gini $\nu=6$ & Variance \\\hline\hline
Axis 1&7.50 & 7.86 & 8.46 & \textbf{11.92} \\
Axis 2&0.62& 0.75 & 0.77 & \textbf{11.24} \\\hline 
\end{tabular}
\caption{Standard deviation of the MSE of the ACT on the two first axes}\label{tab2}
\end{center}
\end{table}

As in the previous example of simulation, Table \ref{tab2} indicates that the PCA based on the variance is less stable about the values of the ACTs that provide the most important observations of the sample. This may lead to irrelevant interpretations.

\section{Application on cars data}\label{s7}

We propose a simple application with the celebrated cars data (see the Appendix).\footnote{An R markdown for Gini PCA is available: \scriptsize{\texttt{\textcolor{blue}{https://github.com/freakonometrics/GiniACP/}}}
\\ Data from Michel Tenenhaus's website (see also the Appendix):\\
\scriptsize{\texttt{\textcolor{blue}{https://studies2.hec.fr/jahia/webdav/site/hec/shared/sites/tenenhaus/acces\_anonyme/home/fichier\_excel/auto\_2004.xls}}} } The dataset is particularly interesting since there are highly correlated variables as can be seen in the Pearson correlation matrix given in Table \ref{tab:corr}.

\begin{table}[h!]
\begin{tabular}{| c || c | c | c | c | c | c |}
\hline   & \textbf{capacity} $x_1$ & \textbf{power} $x_2$ &  \textbf{speed} $x_3$ &  \textbf{weight} $x_4$ & \textbf{width} $x_5$ & \textbf{length} $x_6$ \\\hline
\hline $x_1$ & \textbf{1.000} & 0.954 & 0.885 & 0.692 & 0.706 & 0.663\\
\hline $x_2$ & 0.954 &\textbf{1.000} & 0.933 & 0.528 & 0.729 & 0.663 \\
\hline $x_3$ & 0.885 & 0.933& \textbf{1.000} & 0.466 & 0.618 & 0.578 \\
\hline $x_4$ & 0.692 & 0.528 & 0.466 & \textbf{1.000} & 0.477 & 0.794 \\
\hline $x_5$  & 0.706 & 0.729 & 0.618 & 0.477 & \textbf{1.000} & 0.591 \\
\hline $x_6$ & 0.663 & 0.663 & 0.578 & 0.794 & 0.591 & \textbf{1.000} \\
\hline 
\end{tabular}
\caption{Correlation matrix}\label{tab:corr}
\end{table}

Also, the dataset is composed of some outlying observations (Figure \ref{fig5}): Ferrari enzo ($x_1$, $x_2$, $x_5$), Bentley continental ($x_2$), Aston Martin ($x_2$), Land Rover discovery ($x_5$), Mercedes class S ($x_5$), Smart ($x_5, x_6$).

\begin{figure}[h!]
\begin{tabular}{ccc}
capacity $x_1$  &  power $x_2$ &  speed $x_3$ \\
\includegraphics[scale=0.17]{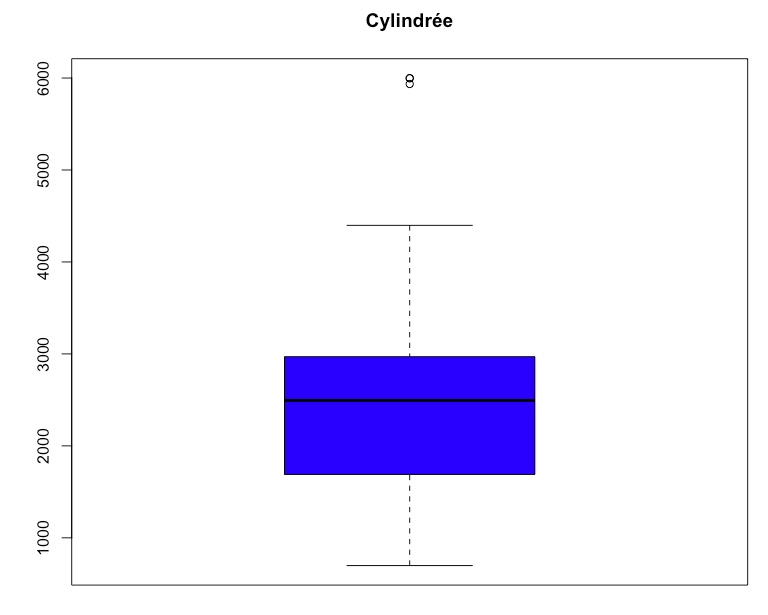}  & \includegraphics[scale=0.17]{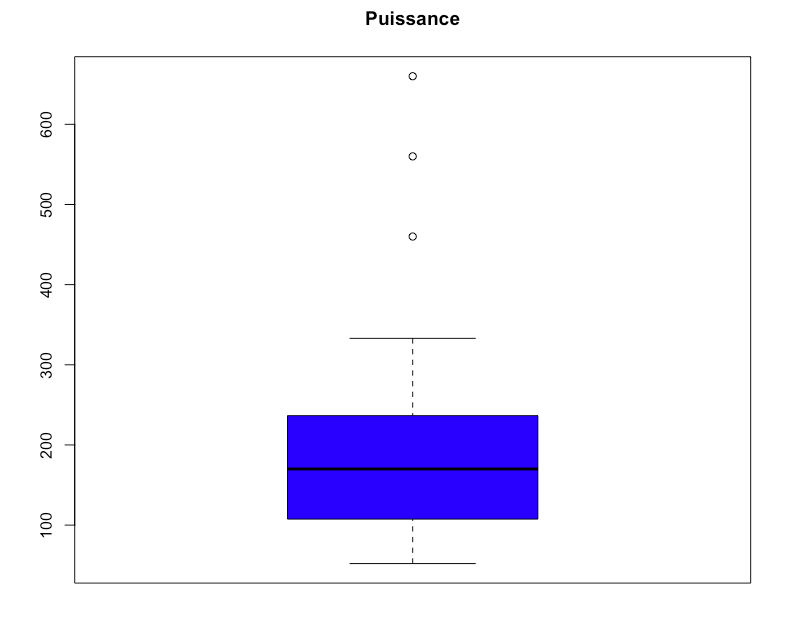} &   \includegraphics[scale=0.17]{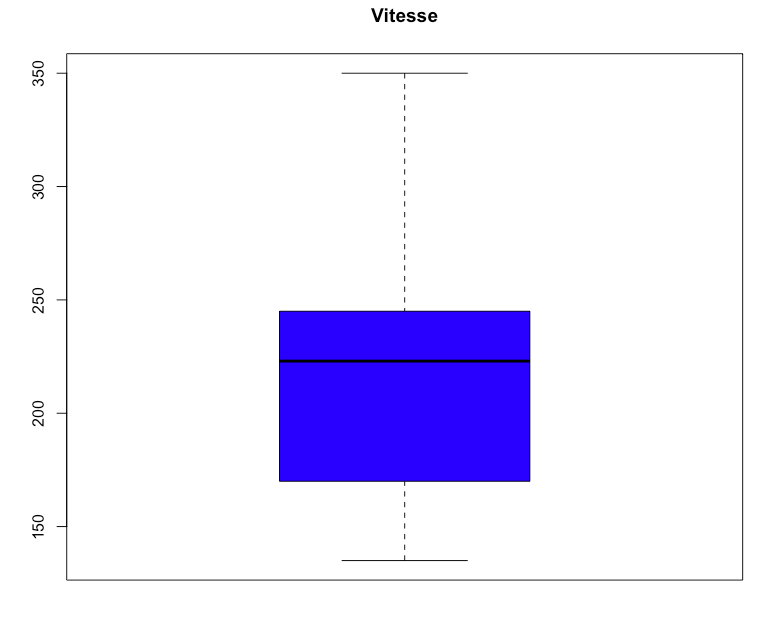} \\
weight $x_4$ &  width $x_5$  &  length $x_6$ \\
\includegraphics[scale=0.17]{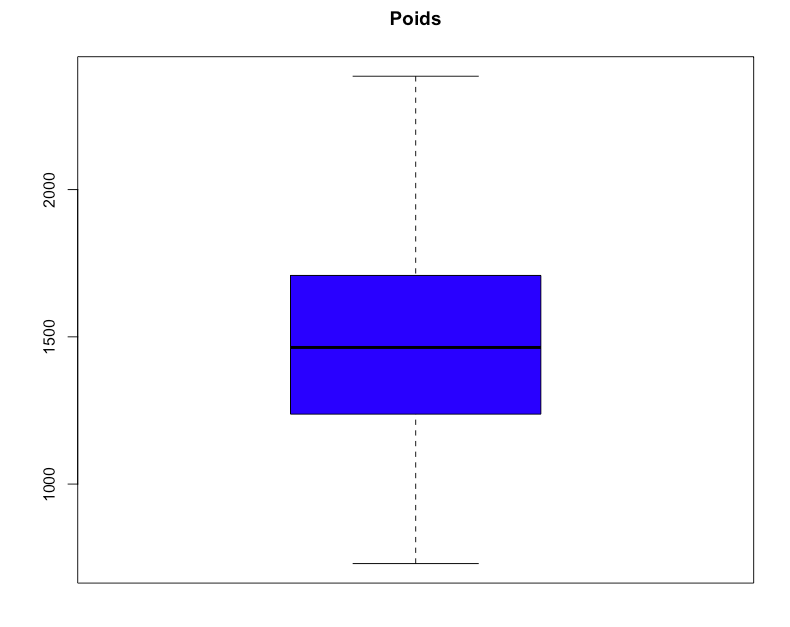} & \includegraphics[scale=0.17]{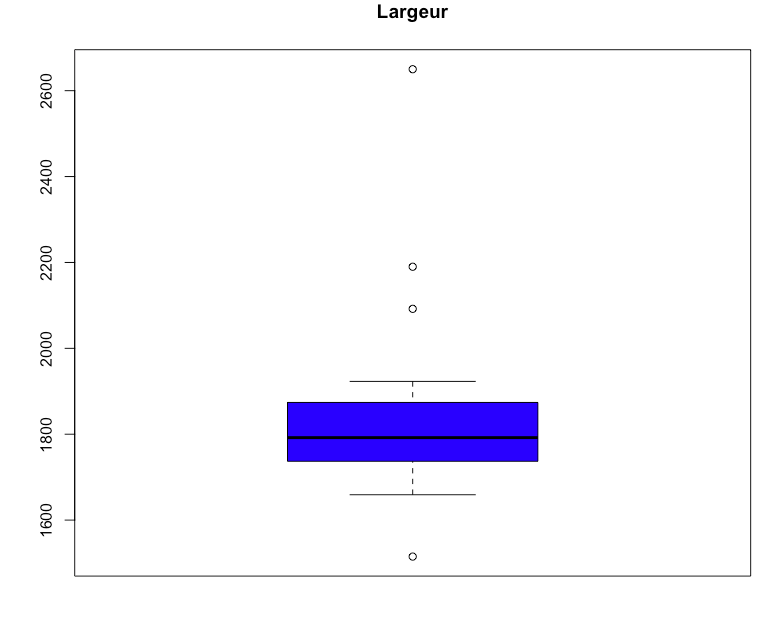}  & \includegraphics[scale=0.17]{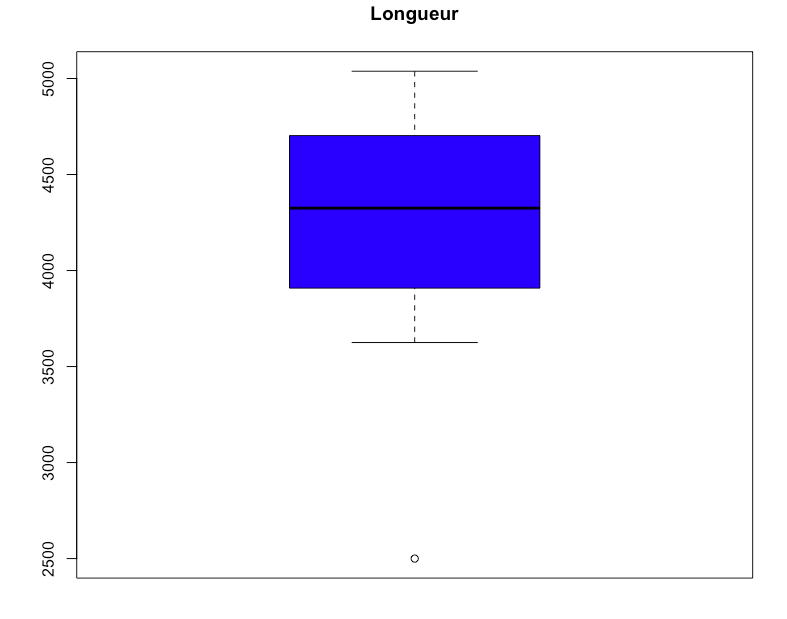}
\end{tabular} 
\caption{Box plots}\label{fig5}
\end{figure}

The overall information (variability) is partitioned over six components (Table \ref{tab6}).

\begin{table}[h!]
\begin{tabular}{| c || c | c | c | c |}
\hline \textbf{eigenvalues in \%}  & Gini $\nu=2$ &  Gini $\nu=4$ & Gini $\nu=6$ & Variance  \\ \hline\hline
\textbf{Axis 1} & \textbf{80.35797} &  \textbf{83.17172}  & \textbf{84.84995} &\textbf{73.52112}\\
\hline \textbf{Axis 2} & \textbf{12.0761}  &  \textbf{10.58655} &  \textbf{9.715974} & \textbf{14.22349}\\ 
\hline  \textbf{Axis 3} & 4.132136 &  2.987015 & 3.130199 & \textbf{7.26106} \\
\hline \textbf{Axis 4}  & 3.059399 & 2.612411  & 1.519626 & 3.93117 \\
\hline  \textbf{Axis 5} & 0.3332362 &  0.3125735  & 0.2696611 & 0.85727 \\
\hline  \textbf{Axis 6} & 0.04115858 &  $- 0.3297257$ & -0.5145944 & 0.20585\\
\hline  \textbf{Sum}  & \textbf{100 \%} &  \textbf{100 \%} & \textbf{100 \%}  & \textbf{100 \%} \\
 \hline
\end{tabular}
\caption{Eigenvalues (\%)}\label{tab6}
\end{table}

Two axes may be chosen to analyze the data. As shown in the previous Section about the simulations, when the data are highly correlated such that two axes are sufficient to project the data, the Gini PCA and the standard PCA yield the same share of information on each axis. However, we can expect some differences for absolute contributions $ACT$ and relative contributions $RCT$. 

The projection of the data is depicted in Figure \ref{fig6}, for each method. 

\begin{figure}[h!]
\begin{subfigure}{.6\textwidth}
\includegraphics[width=1\linewidth]{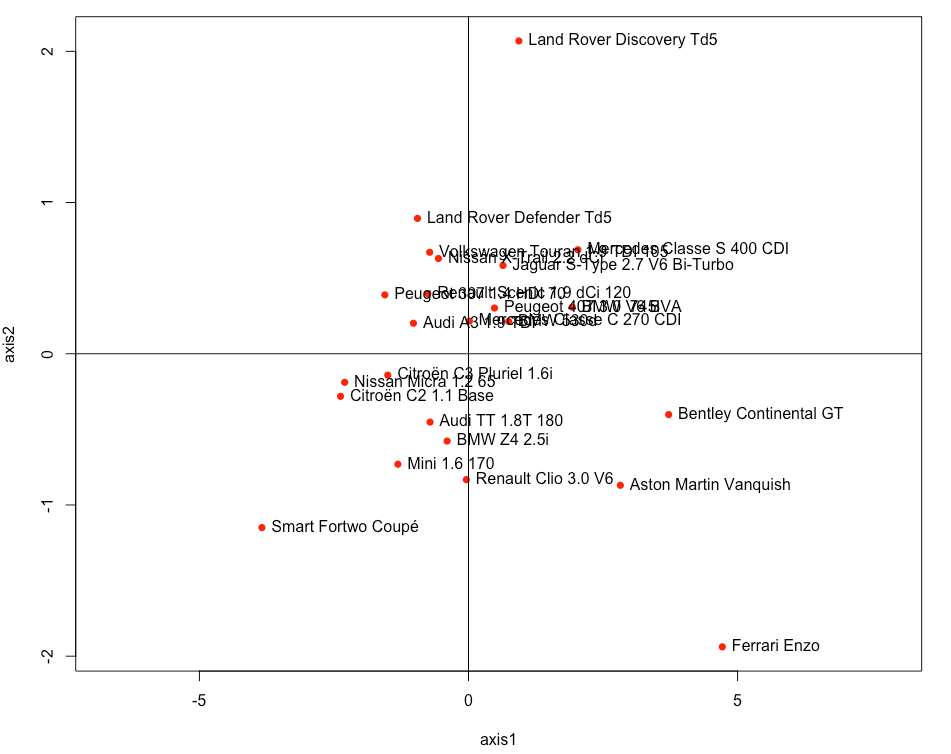} 
\caption{Gini $(\nu = 2)$}
\end{subfigure}%
\begin{subfigure}{.6\textwidth}
\includegraphics[width=1\linewidth]{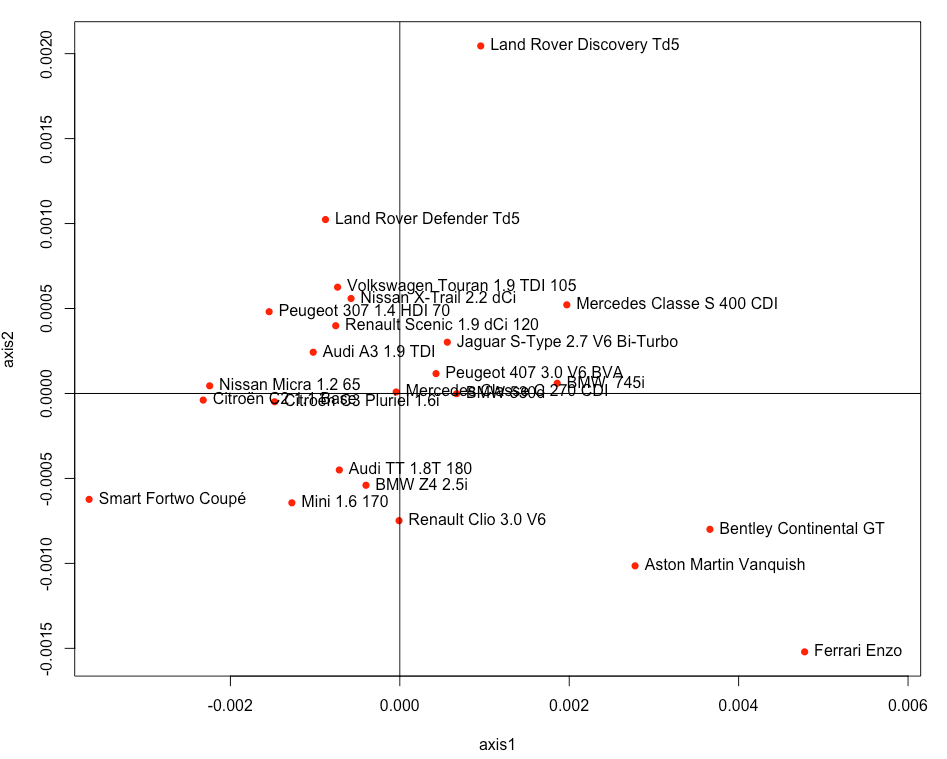} 
\caption{Gini $(\nu = 4)$}
%\textbf{Figure 3a:} $ACT_1$ &  \textbf{Figure 3b:} $ACT_2$ \\
\end{subfigure}%

\begin{subfigure}{.6\textwidth}
\includegraphics[width=1\linewidth]{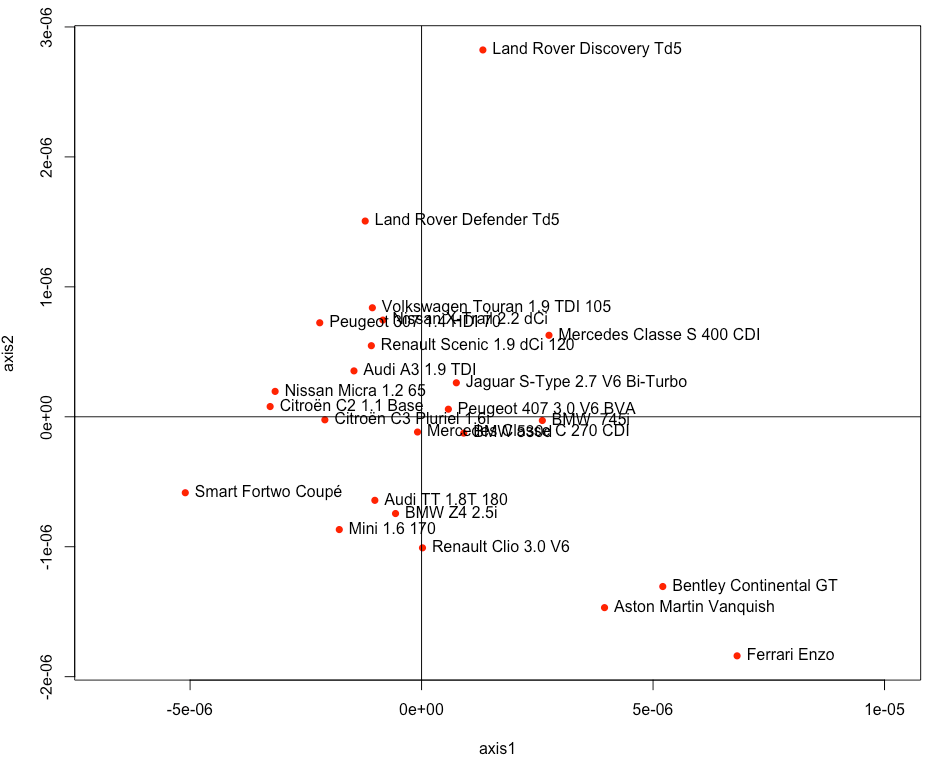} 
\caption{Gini $(\nu = 6)$}
\end{subfigure}%
\begin{subfigure}{.6\textwidth}
\includegraphics[width=1\linewidth]{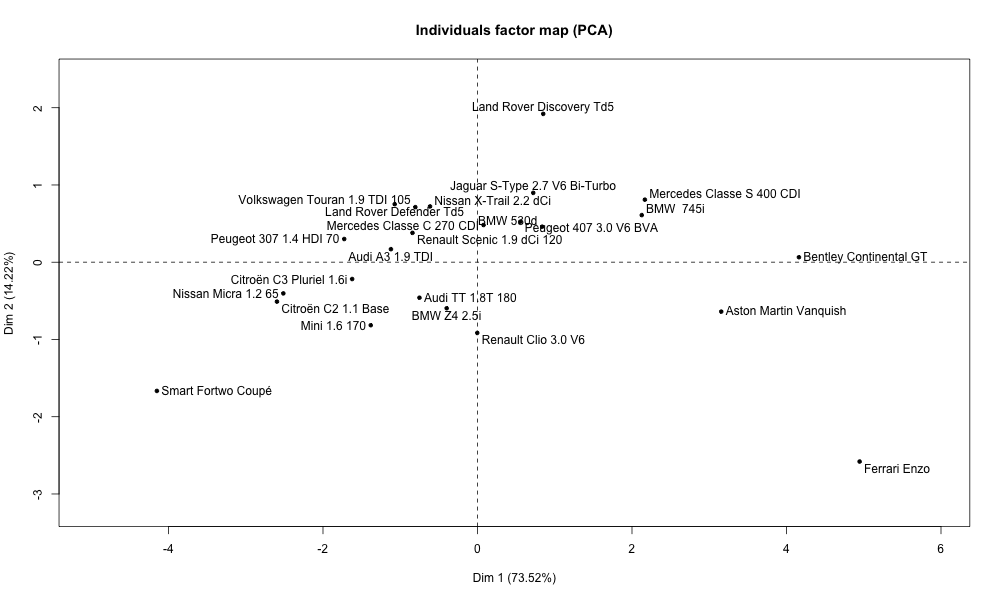} 
\caption{Variance}
%\textbf{Figure 3c:} $RCT_1$ &  \textbf{Figure 3d:} $RCT_2$ 
\end{subfigure}%
\caption{Projections of the cars}\label{fig6}
\end{figure}

% \begin{figure}
% \begin{tabular}{cc}
% Gini $(\nu = 2)$   &  Gini $(\nu = 4)$    \\
% \includegraphics[width=.8\linewidth]{Proj_ind_nu2.png}  & \includegraphics[width=.8\linewidth]{Proj_ind_nu4.png}  \\
% Gini $(\nu = 6)$   &  variance    \\
% \includegraphics[width=.8\linewidth]{Proj_ind_nu6.png}  & \includegraphics[width=.8\linewidth]{Proj_ind_var.png}  \\
% \end{tabular} 
% \caption{Projections of the cars}
% \end{figure}
As depicted in Figure \ref{fig6}, the projection is very similar for each technique. The cars with extraordinary (or very low) abilities are in the same relative position in the four projections: Land Rover Discovery Td5 at the top, Ferrari Enzo at the bottom right, Smart Fortwo coup\'e at the bottom left. However, when we improve the coefficient of variability $\nu$ to look for what happens at the tails of the distributions (of the two axes), we see that more cars are distinguishable: Land Rover Defender, Audi TT, BMW Z4, Renault Clio 3.0 V6, Bentley Contiental GT. Consequently, contrary to the case $\nu=2$ or the variance, the projections with $\nu=4,6$ allow one to find other important observations, which are not outlying observations that contribute to the overall amount of variability. For this purpose, let us first analyze the correlations between the variables and the new axes in order to interpret the results, see Tables \ref{tab7} to \ref{tab10}.     

Some slight differences appear between the Gini PCA and the classical one based on the variance. The theoretical Section \ref{s4} indicates that the Gini methodology for $\nu=2$ is equivalent to the variance when the variables are Gaussian. On cars data, we observe this similarity. In each PCA, all variables are correlated with Axis 1 and weight with Axis 2. However, when $\nu$ increases, the Gini methodology allows outlying observations to be diluted so that some variables may appear to be significant, whereas they are not in the variance case. 

\begin{table}[h!]
\begin{tabular}{| c | c || l | l | l | l | l | c |}
\hline \multicolumn{2}{|c||} {\textbf{Gini} ($\nu = 2$)}  & capacity &  power & speed & weight & width & length \\
\hline
\hline  
            & correlation & -0.974 & -0.945  & -0.872 & 0.760 & -0.933 & -0.823 \\
 \cline{3-8}   Axe 1 & &   &  &   &  & & \\
           & \textbf{$U$-stat} & \cellcolor{lightgray} \textbf{ -56.416} & \cellcolor{lightgray} \textbf{-25.005} & \cellcolor{lightgray} \textbf{-10.055} & \cellcolor{lightgray} \textbf{-4.093} & \cellcolor{lightgray} \cellcolor{lightgray}\textbf{-24.837} &\cellcolor{lightgray} \textbf{-12.626} \\
\hline 
\hline 
& correlation & -0.032 & -0.241  & -0.405 & 0.510 & 0.183 & -0.379 \\
 \cline{3-8}   Axe 2 & &   &  &   &  & & \\
           & \textbf{$U$-stat} & \textbf{ -0.112} & \textbf{-0.920} & \textbf{-1.576} & \cellcolor{lightgray} \textbf{2.897}  & \textbf{0.526} & \textbf{1.666} \\
\hline 
\end{tabular}
\caption{Correlations Axes / variables (significance \textcolor{gray}{5\%})}\label{tab7}
\end{table}

\begin{table}[h!]
\begin{tabular}{| c | c || l | l | l | l | l | c |}
\hline \multicolumn{2}{|c||} {\textbf{Gini} ($\nu = 4$)}  & capacity &  power & speed & weight & width & length \\
\hline
\hline  
            & correlation & 0.982 & 0.948  & 0.797 & 0.858 & 0.952 & 0.888 \\
 \cline{3-8}   Axe 1 & &   &  &   &  & & \\
           & \textbf{$U$-stat} & \cellcolor{lightgray} \textbf{ 8.990} & \cellcolor{lightgray} \textbf{8.758} & \cellcolor{lightgray} \textbf{4.805} & \cellcolor{lightgray} \textbf{9.657} & \cellcolor{lightgray} \cellcolor{lightgray}\textbf{8.517} &\cellcolor{lightgray} \textbf{8.182} \\
\hline 
\hline 
& correlation & -0.021 & 0.207  & 0.516 & -0.279 & -0.147 & -0.246 \\
 \cline{3-8}   Axe 2 & &   &  &   &  & & \\
           & \textbf{$U$-stat} & \textbf{ -0.095} & \textbf{0.817} & \cellcolor{lightgray}  \textbf{2.299} &  \textbf{-1.773}  & \textbf{-0.705} & \textbf{-1.200} \\
\hline 
\end{tabular}
\bigskip

\caption{Correlations Axes / variables (significance \textcolor{gray}{5\%})}\label{tab8}
\end{table}

\begin{table}[h!]
\begin{tabular}{| c | c || l | l | l | l | l | c |}
\hline \multicolumn{2}{|c||} {\textbf{Gini} ($\nu = 6$)}  & capacity &  power & speed & weight & width & length \\
\hline
\hline  
            & valeurs & -0.781 & -0.759  & -0.598 & -0.730 & -0.755 & -0.701 \\
 \cline{3-8}   Axe 1 & &   &  &   &  & & \\
           & \textbf{$U$-stat} & \cellcolor{lightgray} \textbf{ -4.036} & \cellcolor{lightgray} \textbf{-3.903} &\cellcolor{lightgray} \textbf{-3.137} & \cellcolor{lightgray} \textbf{-3.125} & \cellcolor{lightgray} \textbf{-3.882} & \cellcolor{lightgray} \cellcolor{lightgray}\textbf{-3.644}  \\
\hline 
\hline 
& valeurs & 0.019 & -0.170  & -0.570 & 0.153 & 0.125 & 0.218 \\
 \cline{3-8}   Axe 2 & &   &  &   &  & & \\
           & \textbf{$U$-stat} & \textbf{ 0.089} & \textbf{-0.734} &  \cellcolor{gray!15}\textbf{-1.914} &  \textbf{0.734}  & \textbf{0.569} & \textbf{0.906} \\
\hline 
\end{tabular}
\bigskip

\caption{Correlations Axes / variables (significance \textcolor{gray}{5\%}, \textcolor{lightgray}{10\%})}\label{tab9}
\end{table}

\begin{table}
\begin{tabular}{| c | c || l | l | l | l | l | c |}
\hline \multicolumn{2}{|c||} {\textbf{Variance}}  & capacity &  power & speed & weight & width & length \\
\hline
\hline  
            & valeurs & 0.962 & 0.923  & 0.886 & 0.756 & 0.801 & 0.795 \\
 \cline{3-8}   Axe 1 & &   &  &   &  & & \\
           & \textbf{$U$-stat} & \cellcolor{lightgray} \textbf{11.802} & \cellcolor{lightgray} \textbf{11.322} &\cellcolor{lightgray} \textbf{10.866} & \cellcolor{lightgray} \textbf{9.282} & \cellcolor{lightgray} \textbf{9.825} & \cellcolor{lightgray} \cellcolor{lightgray}\textbf{9.752}  \\
\hline 
\hline 
& valeurs & -0.126 & -0.352  & -0.338 & 0.575 & -0.111 & 0.504 \\
 \cline{3-8}   Axe 2 & &   &  &   &  & & \\
           & \textbf{$U$-stat} & \textbf{ -0.307} & \textbf{-0.855} &  \textbf{-0.821} &  \cellcolor{gray!15}\textbf{1.396}  & \textbf{-0.269} & \textbf{1.223} \\
\hline 
\end{tabular}

\bigskip
\caption{Correlations Axes / variables (significance \textcolor{gray}{5\%}, \textcolor{lightgray}{10\%})}\label{tab10}
\end{table}

Tables \ref{tab8} and \ref{tab9} ($\nu=4,6$) show that Axis 2 is correlated to speed (not weight as in the variance PCA). In this respect the absolute contributions must describe the cars associated with speed on Axis 2. Indeed, the Land Rover discovery, a heavy weight car, is no more available on Axis 2 for the Gini PCA for $\nu=2,4,6$ (Figures \ref{fig:8}, \ref{fig:9}, \ref{fig:10}). Note that the red line in the Figures represents the mean share of the information on each axis, \emph{i.e.} 100\%/24 cars = 4.16\% of information per car.

\begin{figure}[h!]
\begin{subfigure}{.6\textwidth}
\includegraphics[width=1\linewidth]{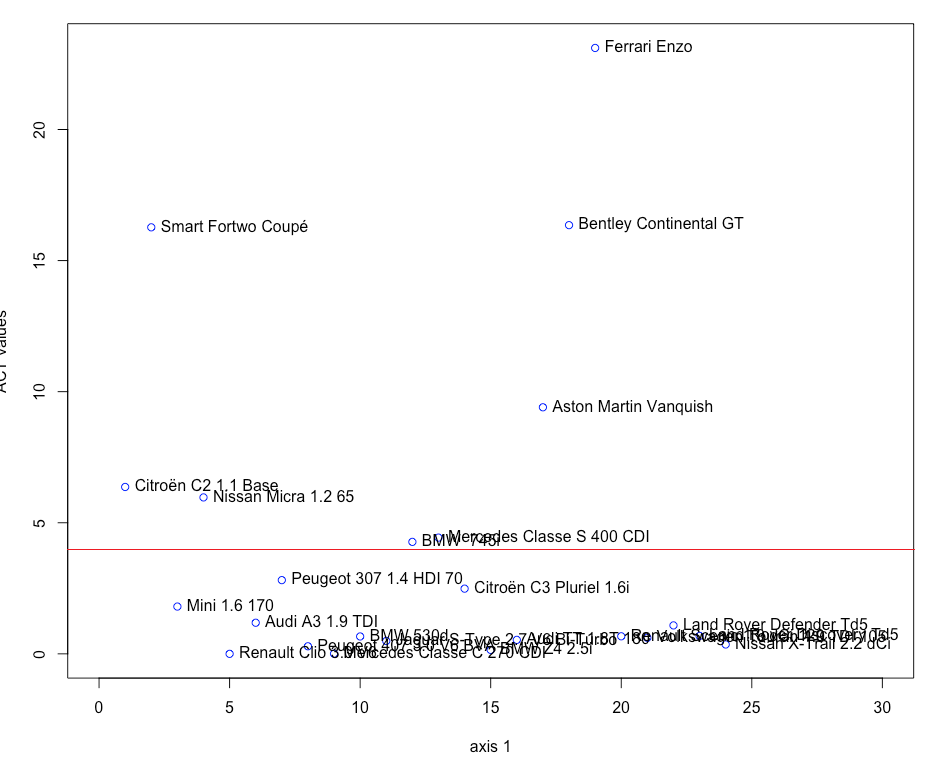} 
\caption{Axis 1 (variance) }
\end{subfigure}%
\begin{subfigure}{.6\textwidth}
\includegraphics[width=1\linewidth]{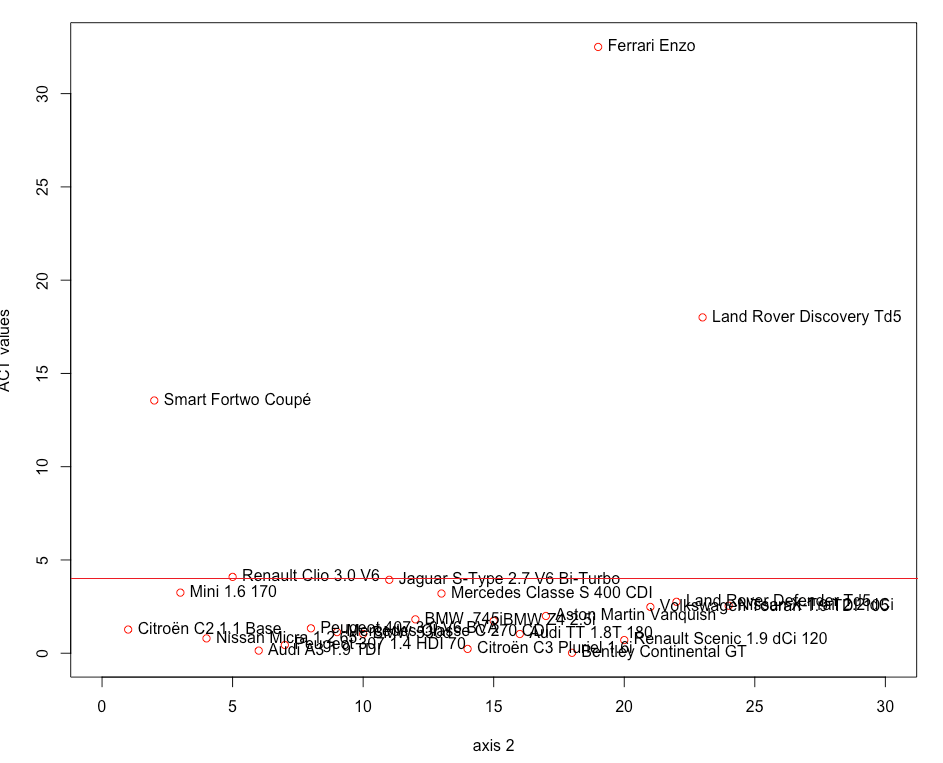} 
\caption{Axis 2 (variance) }
%\textbf{Figure 3a:} $ACT_1$ &  \textbf{Figure 3b:} $ACT_2$ \\
\end{subfigure}%
\caption{Variance ACTs}\label{fig:7}
\end{figure}

\begin{figure}[h!]
\begin{subfigure}{.6\textwidth}
\includegraphics[width=1\linewidth]{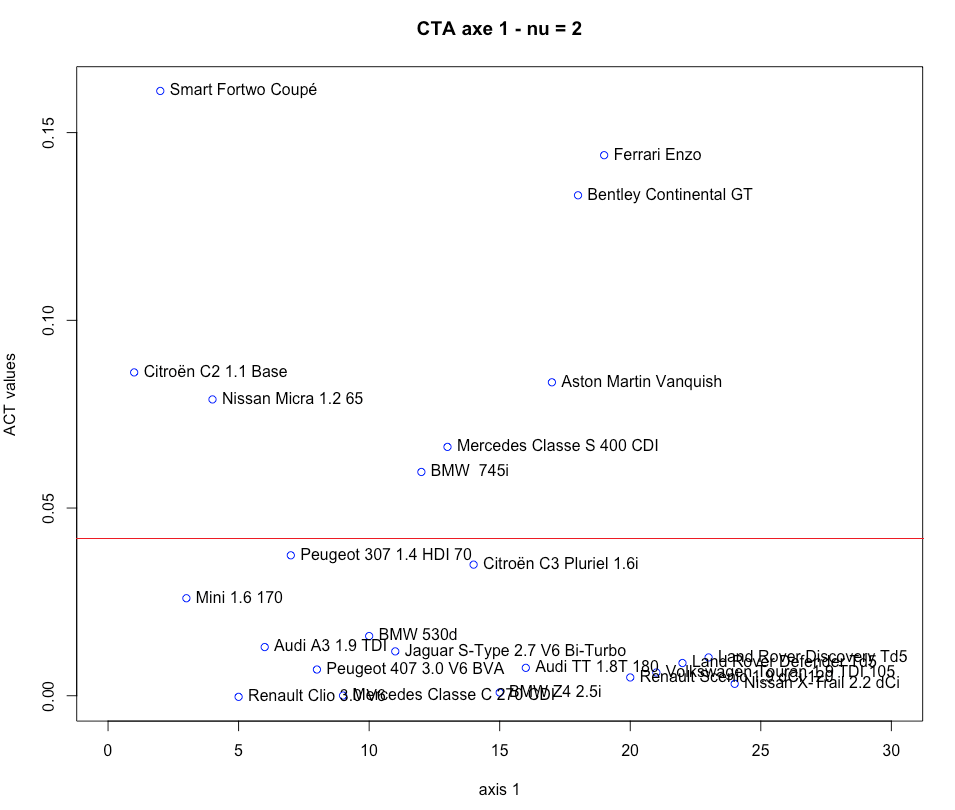} 
\caption{Axis 1 ($\nu=2$) }
\end{subfigure}%
\begin{subfigure}{.6\textwidth}
\includegraphics[width=1\linewidth]{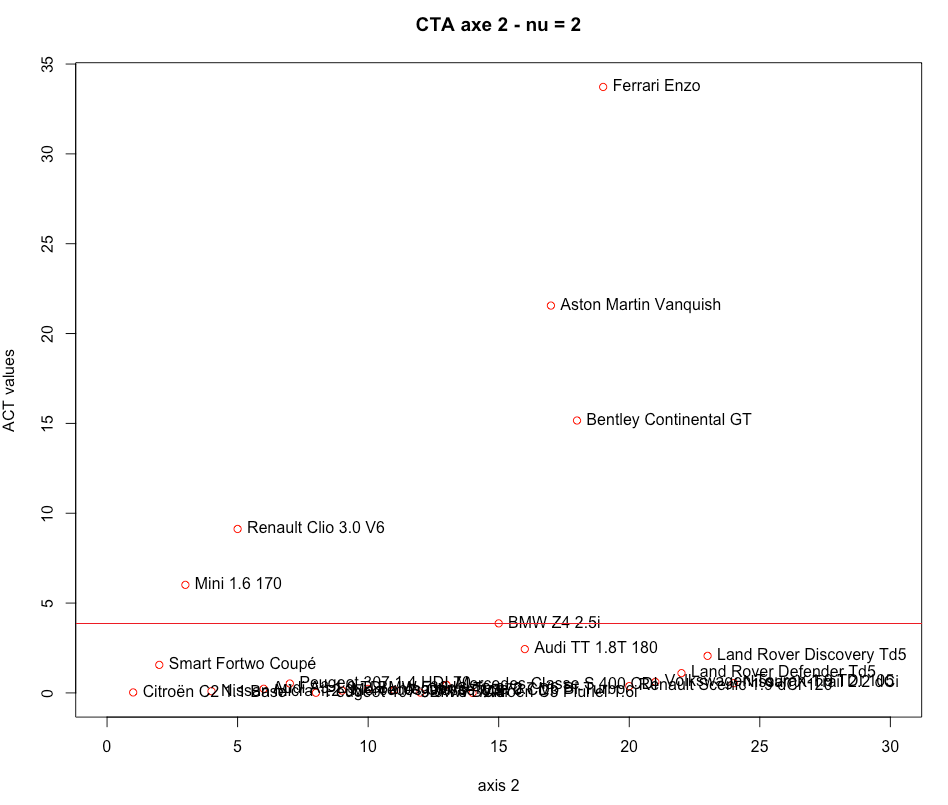} 
\caption{Axis 2 ($\nu=2$) }
%\textbf{Figure 3a:} $ACT_1$ &  \textbf{Figure 3b:} $ACT_2$ \\
\end{subfigure}%
\caption{Gini ACTs ($\nu=2$)}\label{fig:8}
\end{figure}

\begin{figure}[h!]
\begin{subfigure}{.6\textwidth}
\includegraphics[width=1\linewidth]{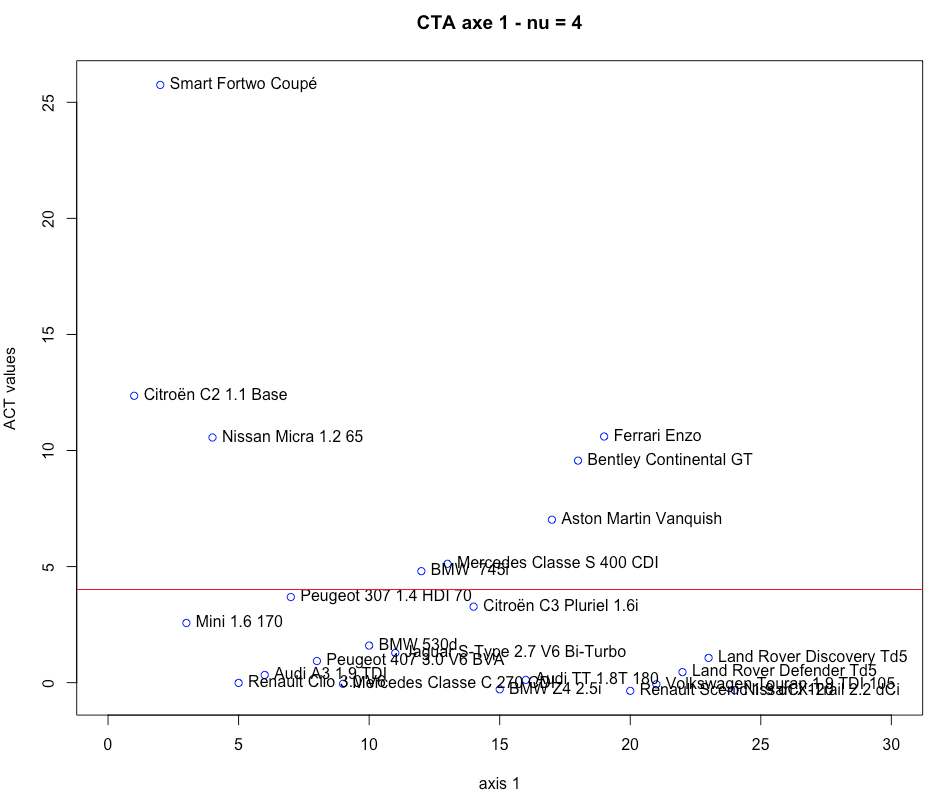} 
\caption{Axis 1 ($\nu=4$) }
\end{subfigure}%
\begin{subfigure}{.6\textwidth}
\includegraphics[width=1\linewidth]{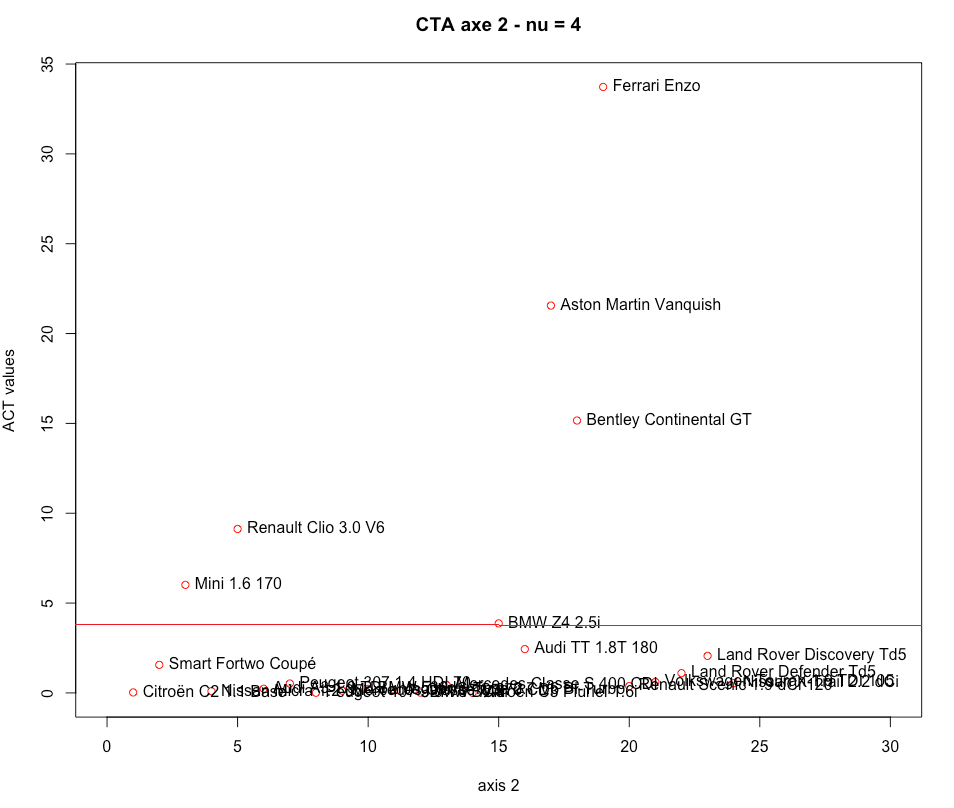} 
\caption{Axis 2 ($\nu=4$) }
%\textbf{Figure 3a:} $ACT_1$ &  \textbf{Figure 3b:} $ACT_2$ \\
\end{subfigure}%
\caption{Gini ACTs ($\nu=4$)}\label{fig:9}
\end{figure}

\begin{figure}[h!]
\begin{subfigure}{.6\textwidth}
\includegraphics[width=1\linewidth]{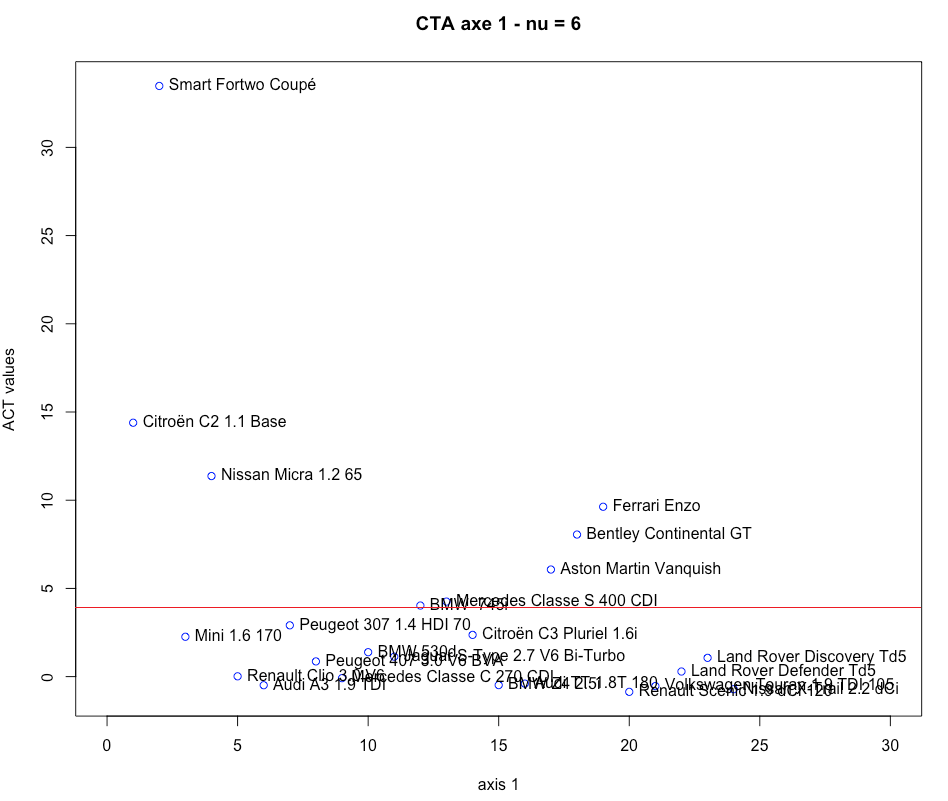} 
\caption{Axis 1 ($\nu=6$) }
\end{subfigure}%
\begin{subfigure}{.6\textwidth}
\includegraphics[width=1\linewidth]{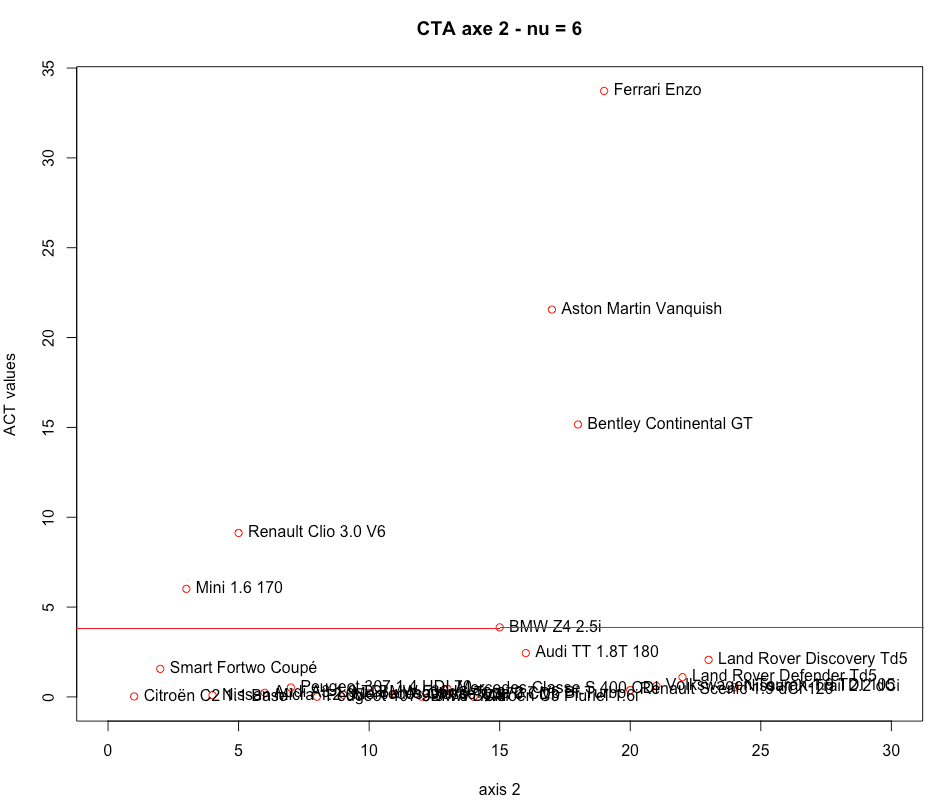} 
\caption{Axis 2 ($\nu=6$) }
%\textbf{Figure 3a:} $ACT_1$ &  \textbf{Figure 3b:} $ACT_2$ \\
\end{subfigure}%
\caption{Gini ACTs ($\nu=6$)}\label{fig:10}
\end{figure}

Finally, some cars are not correlated with axis 2 in the standard PCA, see Figures \ref{fig:8}--\ref{fig:10}, while this is the case in the Gini PCA. Indeed some cars are now associated with speed: Aston Martin, Bentley Continental GT, Renault Clio 3.0 V6 and Mini 1.6 170. This example of application shows that the use of the Gini metric robust to outliers may involve some serious changes in the interpretation of the results. 

\section{Conclusion}

In this paper, it has been shown that the geometry of the Gini covariance operator allows one to perform Gini PCA, that is, a robust principal component analysis based on the $\ell_1$ norm.

\medskip

To be precise, the variance may be replaced by the Gini Mean Difference, which captures the variability of couples of variables based on the rank of the observations in order to attenuate the influence of the outliers. The Gini Mean Difference may be rather interpreted with the aid of the generalized Gini index $GGMD_\nu$ in the new subspace for a better understanding of the variability of the components, that is, $GGMD_\nu$ is both a rank-dependent measure of variability in \cite{Yaari1987} sense and also an eigenvalue of the Gini correlation matrix. 

\medskip

Contrary to many approaches in multidimensional statistics in which the standard variance-covariance matrix is used to project the data onto a new subspace before deriving multidimensional Gini indices (see e.g. \cite{Banerjee}), we propose to employ the Gini correlation indices (see \cite{Yitzhaki13}). This provides the ability to interpret the results with the $\ell_1$ norm and the use of $U$-statistics to measure the significance of the correlation between the new axes and the variables.  

\medskip

This research may open the way on data analysis based on Gini metrics in order to study multivariate correlations with categorical variables or discriminant analyses when outlying observations drastically affect the sample.

\break
\section*{Appendix}

\begin{table}[h]
\begin{footnotesize}
\begin{tabular}{|c ||c c c c c c|} 
\hline
\hline
cars& capacity $x_1$ &power $x_2$ &speed $x_3$ &weight $x_4$ &width $x_5$  & length $x_6$\\ \hline 
Citro\"{e}n C2 1.1 Base	&1124	&61	&158	&932	&1659	&3666 \\
Smart Fortwo Coup\'{e}&	698	&52	&135	&730	&1515	&2500 \\
Mini 1.6 170	&1598	&170	&218	&1215	&1690	&3625 \\
Nissan Micra 1.2 65	&1240	&65	&154	&965	&1660	&3715 \\
Renault Clio 3.0 V6	&2946	&255	&245	&1400	&1810	&3812 \\
Audi A3 1.9 TDI	&1896	&105	&187	&1295	&1765	&4203\\
Peugeot 307 1.4 HDI 70	&1398	&70	&160	&1179	&1746	&4202\\
Peugeot 407 3.0 V6 BVA	&2946	&211	&229	&1640	&1811	&4676\\
Mercedes Classe C 270 CDI	&2685	&170	&230	&1600	&1728	&4528\\
BMW 530d&	2993	&218	&245	&1595	&1846	&4841\\
Jaguar S-Type 2.7 V6 Bi-Turbo&	2720	&207	&230	&1722	&1818	&4905\\
BMW  745i	&4398	&333	&250	&1870	&1902	&5029\\
Mercedes Classe S 400 CDI	&3966	&260	&250	&1915	&2092	&5038\\
Citro\"{e}n C3 Pluriel 1.6i	&1587	&110	&185	&1177	&1700	&3934\\
BMW Z4 2.5i	&2494	&192	&235	&1260	&1781	&4091\\
Audi TT 1.8T 180	&1781	&180	&228	&1280	&1764	&4041\\
Aston Martin Vanquish	&5935	&460	&306	&1835	&1923	&4665\\
Bentley Continental GT	&5998	&560	&318	&2385	&1918	&4804\\
Ferrari Enzo	&5998	&660	&350	&1365	&2650	&4700\\
Renault Scenic 1.9 dCi 120	&1870	&120	&188	&1430	&1805	&4259\\
Volkswagen Touran 1.9 TDI 105	&1896	&105	&180	&1498	&1794	&4391\\
Land Rover Defender Td5	&2495	&122	&135	&1695	&1790	&3883\\
Land Rover Discovery Td5	&2495	&138	&157	&2175	&2190	&4705\\
Nissan X-Trail 2.2 dCi	&2184	&136	&180	&1520	&1765	&4455\\
\hline
\end{tabular}
\end{footnotesize}
\caption{Cars data}
\end{table}


\begin{thebibliography}{99}
\bibitem[Abramowitz \& Stegun(1964)]{Abramowitz} Abramowitz, M. \& I. Stegun. (1964). \emph{Handbook of Mathematical Functions with Formulas, Graphs, and Mathematical Tables}. National Bureau of Standards Applied Mathematics Series No. 55.

\bibitem[Anderson(1963)]{Anderson} Anderson, T.W. (1963) Asymptotic Theory for Principal Component Analysis. {\em The Annals of Mathematical Statistics}, {\bf 34}, 122--148. 

\bibitem[Baccini, Besse \& de Falguerolles(1996)]{Baccini} Baccini, A., P. Besse \& A. de Falguerolles (1996), A $L_1$ norm PCA and a heuristic approach, in \emph{Ordinal and Symbolic Data Analysis}, E Didday, Y. Lechevalier and O. Opitz (eds), Springer, 359--368.

\bibitem[Banerjee(2010)]{Banerjee} Banerjee, A.K. (2010), A multidimensional Gini index, \emph{Mathematical Social Sciences}, {\bf 60}: 87--93. 

\bibitem[Candes {\em et al.}(2009)]{Candes} Candes, E.J., Xiaodong Li, Yi Ma \& John Wright. (2009) Robust Principal Component Analysis?. arXiv:0912.3599.

\bibitem[Carcea and Serfling(2015)]{ref2}
 Carcea, M. \& R. Serfling (2015), A Gini autocovariance function for time series modeling. \emph{Journal of Time Series Analysis}  \textbf{36}: 817--38.

\bibitem[Dalton(1920)]{Dalton1920} Dalton, H. 1920. The Measurement of the Inequality of Incomes. {\em The Economic Journal}, 
{\bf 30}:119, 348--361.

\bibitem[d'Aspremont {\em et al.}(1920)]{dAspremont} d'Aspremont, A., L. El Ghaoui, M.I. Jordan, \& G. R. G. Lanckriet (2007).  A direct formulation for sparse PCA using semidefinite programming. {\em SIAM Review}, {\bf 49}:3, 434--448.

\bibitem[Decancq \& Lugo(2013)]{Decancq} Decancq, K. \& M.-A. Lugo (2013), Weights in Multidimensional Indices of Well-Being: An Overview, \emph{Econometric Reviews}, {\bf 32}:1, 7--34.
 
\bibitem[Ding {\em et al.}(2006)]{Ding2006} Ding, C., Zhou, D., He, X. \& Zha, H.  (2006). $R_1$-PCA: rotational invariant L1-norm principal component analysis for robust subspace factorization. {\em ICML '06 Proceedings of the 23rd international conference on Machine learning}, 281--288

\bibitem[Eckart \& Young(1936)]{Eckart} Eckart, C. \& G. Young (1936)  The approximation of one matrix by another of lower rank. {\em Psychometrika},{\bf 1}, 211--218.

\bibitem[Flury \& Riedwyl(1988)]{FluryRiedwyl} Flury \& Riedwyl (1988).  Multivariate Statistics: A Practical Approach. Chapman \& Hall

\bibitem[Furman \& Zitikis(2017)]{furman} Furman, E. \& R. Zitikis (2017), Beyond the Pearson Correlation: Heavy-Tailed Risks, Weighted Gini Correlations, and A Gini-Type Weighted Insurance Pricing Model, \emph{ASTIN Bulletin: The Journal of the International Actuarial Association}, \textbf{47(03)}: 919-942. 

\bibitem[Gajdos \& Weymark(2005)]{Gajdos} Gajdos, T. and J. Weymark (2005), Multidimensional generalized Gini indices, \emph{Economic Theory}, {\bf 26}:3, 471-496.

\bibitem[Gini(1912)]{Gini12} Gini, C. (1912), Variabilit\`{a} e mutabilit\`{a}, Memori di Metodologia Statistica, Vol. 1, Variabilit\`{a} e Concentrazione. Libreria Eredi Virgilio Veschi, Rome, 211--382.

%\bibitem[Gini(1921)]{Gini1921} Gini, C. (1921). Measurement of Inequality of Incomes.  {\em The
%Economic Journal}, {\bf 31}, 124--126.

\bibitem[Giorgi(2013)]{giorgi} Giorgi, G.M. (2013), Back to the future: some considerations on Shlomo Yitzhaki and Edna Schechtman's book "`The Gini Methodology: A Primer on a Statistical Methodology"', \emph{Metron}, \textbf{71(2)}: 189-195.

\bibitem[Gorban {\em et al.}(2007)]{Gorban} Gorban, A.N. , B. Kegl, D.C. Wunsch, \& A. Zinovyev (Eds.) (2007) Principal Manifolds for Data Visualisation and Dimension Reduction. LNCSE 58, Springer Verlag.

\bibitem[Hotelling(1933)]{Hotelling} Hotelling, H. (1933)   Analysis  of  a  complex  of  statistical  variables  into  principal  components. {\em Journal  of Educational Psychology}, {\bf 24}, 417--441, 1933.

\bibitem[Korhonen \& Siljam\"aki(1998)]{Korhonen} Korhonen, P. \& Siljam\"aki, A. (1998). Ordinal principal component analysis theory and an application. {\em Computational Statistics \& Data Analysis}, {\bf 26}:4, 411--424.

\bibitem[List(1999)]{List} List, C. (1999), Multidimensional Inequality Measurement:
A Proposal, \emph{Working paper}, Nuffield College. 

\bibitem[Mackey(2009)]{Mackey} Mackey, L. (2009)  Deflation methods for sparse PCA. {\em Advances  in Neural  Information Processing Systems}, {\bf 21}: 1017--1024.

\bibitem[Mardia, Kent \& Bibby(1979)]{Mardia} Mardia, K, Kent, J. \& Bibby, J. (1979). Multivariate Analysis. Academic Press, London.

%\bibitem[Mussard \& Souissi-Benrejab(2018)]{Mussard} Mussard, S. \& F. Souissi-Benrejab (2018), Gini-PLS regressions, \emph{Journal of Quantitative Economics}, forthcoming. 

\bibitem[Olkin and Yitzhaki(1992)]{ref9}
 Olkin, Ingram, and Shlomo Yitzhaki. 1992.  Gini regression analysis. \textit{International Statistical Review} \textbf{60}: 185--96.

\bibitem[Pearson(1901)]{Pearson} Pearson, K. (1901), On Lines and Planes of Closest Fit to System of Points in Space, {\em Philosophical Magazine}, {\bf 2}: 559--572.

\bibitem[Saad(1998)]{Saad}  Saad, Y. (1998). Projection and deflation methods for partial pole assignment in linear state feedback .{\em IEEE Trans. Automat. Contr.}, {\bf 33}: 290--297.

\bibitem[Schechtman \& Yitzhaki(1987)]{Schechtman87} Schechtman, E. \& S. Yitzhaki (1987), A  Measure  of  Association Based on Gini's  Mean Difference, \emph{Communications in Statistics: A}, {\bf 16}: {207--231}.


\bibitem[Yitzhaki \& Schechtman(2003)]{Schechtman03} Schechtman, E. \& S. Yitzhaki (2003), A family of correlation coefficients based on the extended Gini index, \emph{Journal of Economic Inequality}, {\bf 1}:2, {129--146}.

\bibitem[Shelef(2016)]{shelef16}Shelef, A. (2016), A Gini-based unit root test. \emph{Computational Statistics \& Data Analysis}, \textbf{100}: 763--772.

\bibitem[Shelef \& Schechtman(2011)]{ref13}
Shelef, A., and E. Schechtman (2011),  A Gini-based methodology for identifying and analyzing time series with non-normal innovations.  \emph{SSNR Electronic Journal} July: 1--26. % please add volume and page.=> OK


\bibitem[Tibshirani(1996)]{Tibshirani} Tibshirani, R. (1996) Regression shrinkage and selection via the LASSO. {\em Journal  of the Royal statistical society, series B}, {\bf 58}:1, 267--288.

\bibitem[Yaari (1987)]{Yaari1987}\ Yaari, M.E. (1987), The Dual Theory of Choice Under Risk, \textit{Econometrica}, {\bf55}: 99--115.

\bibitem[Yaari (1988)]{Yaari1988}\ Yaari, M.E. (1988), A Controversial Proposal Concerning Inequality Measurement, \textit{Journal of Economic Theory}, {\bf 44}: 381--397.

\bibitem[Yitzhaki(1991)]{Yitzhaki1990} Yitzhaki, S. (1991), Calculating Jackknife Variance Estimators
for Parameters of the Gini Method. {\em Journal of Business and Economic Statistics}, {\bf 9}: 235--239.


\bibitem[Yitzhaki(2003)]{Yitzhaki03} Yitzhaki, S. (2003), Gini's Mean difference: a superior measure of variability for non-normal distributions, \emph{Metron}, \textbf{LXI(2)}: 285-316.

\bibitem[Yitzhaki \& Olkin(1991)]{Olkin} Yitzhaki, S. \& Olkin, I. (1991). Concentration indices and concentration curves. {\em Institute of Mathematical Statistics Lecture Notes}, {\bf 19}: 380--392.


\bibitem[Yitzhaki \& Schechtman(2013)]{Yitzhaki13} Yitzhaki, S. \& E. Schechtman (2013), \emph{The Gini Methodology. A Primer on a Statistical Methodology}, Springer.

\bibitem[Zou, Hastie \& Tibshirani(2006)]{Zou} Zou, H., Hastie, T. \& R. Tibshirani (2006), Sparse Principal Component Analysis, \emph{Journal of Computational and Graphical Statistics}, {\bf15}:2, 265-286.


\end{thebibliography}
\end{document}